\documentclass[12pt,lettersize,journal,onecolumn]{IEEEtran}

\usepackage[T1]{fontenc} 
\usepackage{amsmath,amsfonts,amssymb,amsthm,mathtools}
\usepackage{algorithmic}
\usepackage{array}
\usepackage{textcomp}
\usepackage{stfloats}
\usepackage{url}
\usepackage{verbatim}
\usepackage{graphicx}
\usepackage{balance}
\usepackage{setspace}
\usepackage{caption}
\usepackage[font=footnotesize]{subcaption}

\usepackage{tikz}
\usepackage{pgfplots}
\pgfplotsset{compat=1.18}
\usetikzlibrary{patterns, patterns.meta, decorations.pathreplacing, shapes.geometric, calc, arrows.meta, matrix, positioning}

\usepackage{blkarray}
\usepackage{multirow}
\usepackage{arydshln}
\usepackage{makecell}
\usepackage{placeins}
\usepackage{float}
\usepackage{xcolor}
\usepackage{enumitem}
\usepackage{multicol,lipsum,xparse}

\usepackage{cite}
\usepackage{alphalph}
\usepackage{breqn} 
\usepackage{csquotes}

\newtheoremstyle{nonitalic}{3pt}{3pt}{\normalfont}{}{\bfseries}{.}{.5em}{}
\theoremstyle{nonitalic}
\newtheorem{theorem}{Theorem}

\newtheorem{corollary}{Corollary}

\newtheorem{remark}{Remark}
\newtheorem{definition}{Definition}

\newcommand\numberthis{\addtocounter{equation}{1}\tag{\theequation}}
\def\BibTeX{{\rm B\kern-.05em{\sc i\kern-.025em b}\kern-.08em
    T\kern-.1667em\lower.7ex\hbox{E}\kern-.125emX}}

\makeatletter
\newcommand*{\rom}[1]{\expandafter\@slowromancap\romannumeral #1@}
\makeatother

\IEEEoverridecommandlockouts

\begin{document}

\title{Semantics in Actuation Systems: From Age of Actuation to Age of Actuated Information}

\author{Ali Nikkhah, \IEEEmembership{Graduate Student Member, IEEE}, Anthony Ephremides, \IEEEmembership{Life Fellow, IEEE}, and Nikolaos Pappas, \IEEEmembership{Senior Member, IEEE}
\thanks{A. Nikkhah and N. Pappas are with the Department of Computer and Information Science, Linköping University, Linköping, Sweden, email: \{\texttt{ali.nikkhah, nikolaos.pappas\}@liu.se}. A. Ephremides is with the Electrical and Computer Engineering, University of Maryland, College Park, MD, USA, email: \texttt{etony@umd.edu}. \\ This work has been supported in part by the Swedish Research Council (VR), ELLIIT, and the European Union (6G-LEADER, 101192080, and SOVEREIGN, 101131481).\\ A shorter version has been published in \cite{nikkhah2024aoai}.}
}

\maketitle
\begin{abstract}
In this paper, we study the timeliness of actions in communication systems where actuation is constrained by control permissions or energy availability. Building on the Age of Actuation (AoA) metric, which quantifies the timeliness of actions independently of data freshness, we introduce a new metric, the \emph{Age of Actuated Information (AoAI)}. AoAI captures the end-to-end timeliness of actions by explicitly accounting for the age of the data packet at the moment it is actuated. 
We analyze and characterize both AoA and AoAI in discrete-time systems with data storage capabilities under multiple actuation scenarios. The actuator requires both a data packet and an actuation opportunity, which may be provided by a controller or enabled by harvested energy. Data packets may be stored either in a single-packet buffer or an infinite-capacity queue for future actuation. For these settings, we derive closed-form expressions for the average AoA and AoAI and investigate their structural differences.
While AoA and AoAI coincide in instantaneous actuation systems, they differentiate when data buffering is present. Our results reveal counterintuitive regimes in which increasing update or actuation rates degrade action timeliness for both AoA and AoAI.
Moreover, as part of the analysis, we obtain a novel closed-form characterization of the steady-state distribution of a Geo/Geo/1 queue operating under the FCFS discipline, expressed solely in terms of the queue length and the age of the head-of-line packet. The proposed metrics and analytical results provide new insights into the semantics of timeliness in systems where information ultimately serves the purpose of actuation.
\end{abstract}

\section{Introduction}
Recent advances in communication systems have increasingly shifted attention away from traditional performance metrics such as throughput and delay, toward \emph{semantic and goal-oriented notions of performance}, where the value of communicated information is related to its purpose and eventual utilization \cite{kountouris2021semantics}. This perspective underlies the emerging paradigm of \emph{semantic communications}, in which communication efficiency is evaluated in terms of how effectively transmitted information enables decisions, actions, or control objectives rather than merely how reliably bits are delivered \cite{gunduz2022beyond,lu2023semantics}. A prominent example of this shift arises in systems that integrate \emph{sensing, communication, and actuation}, as envisioned in applications such as networked control systems and cyber-physical systems. In such systems, information is valuable only to the extent that it enables \emph{timely actions}. Consequently, traditional metrics that quantify the freshness of information at a receiver do not fully capture the end-to-end timeliness of the overall system when actuation is explicitly involved.

The \emph{Age of Information (AoI)} \cite{pappas2023age} is a well-established metric that has largely replaced earlier timeliness measures in communication networks and control systems. AoI quantifies the freshness of information by measuring the time elapsed since the most recent update was generated at a destination. Although AoI is inherently an information-centric rather than a contextual metric, its analytical tractability makes it a powerful tool and its conceptual robustness a natural foundation for developing semantics-aware notions of timeliness \cite{luo2025informationfreshnesssemanticsinformation}.

In many emerging communication systems as in networked control systems, information is not an end in itself but a means to enable \emph{timely actions}. While substantial effort has been devoted to quantifying the freshness of information delivered to a receiver, the timeliness and relevance of the \emph{actions} triggered by that information have remained largely unexplored. In practical systems, actuation is often constrained by scheduling decisions, controller permissions, or energy availability, leading to data packets being buffered and actuated at a later time. As a result, two actions that occur equally often may differ substantially in their effectiveness depending on the age of the data that triggered them. This reveals a fundamental limitation of information-centric timeliness metrics and motivates the need for action-aware measures that explicitly account for both \emph{when} an action occurs and \emph{how fresh} the actuating information is.

\subsection{Related Work}

A growing body of literature has investigated the role of AoI in control systems involving actuation. In this context, \cite{zhao2019toward} suggests that the co-design of communication and control systems can be made more efficient by explicitly accounting for predictive horizons. In \cite{chang2020age}, the AoI of actuation updates is analyzed by modeling transmission in two stages, from the sensor to the controller and from the controller to the actuator, with the age evolving from the command generation instant to the actuation time while incorporating predictive execution. The trade-off between monetary cost and the AoI of actuation status updates is studied in \cite{kyung2024priority}, where different costs are associated with delivering updates from the controller to the actuator. However, the age is still measured from the controller side. An \emph{effective AoI} metric is proposed in \cite{chang2021effective} for wireless feedback control systems, capturing both the inter-update interval at the source and the delay from sampling to execution. Joint optimization of sampling and scheduling policies is considered in \cite{champatiperformance2019}.

Despite these efforts, the notion of actuation in the above studies remains largely abstract: actuation or execution is typically considered complete upon packet reception, without explicitly accounting for the actuation process itself. This abstraction stems in part from the presence of two distinct phases in such systems: information transmission from the plant (sensor) to the controller, and from the controller to the actuator. While this separation introduces additional complexity, it can be rigorously addressed using queueing-theoretic models and appropriate scheduling and service policies.

A key element that fundamentally enables actuation is the availability of \emph{energy}, or more generally, the presence of an \emph{opportunity} to act. With the proliferation of small-scale and autonomous devices, energy harvesting (EH) has become essential for sustaining long-term system operation \cite{sudevalayam2011energy}. Although EH is widely incorporated into communication system models, it is mainly considered in the context of data transmission rather than actuation.

Within the literature on AoI and EH, two main research directions can be identified. The first considers controlled energy sources, such as wireless power transfer systems \cite{nikkhah2023age,nikkhah2023ageo,krikidis2019average,ibrahim2016stability,abdelmagid2020aoi,li2026aoi,zhang2025aoi}, where energy availability is actively managed as part of the system design. The second direction, to which the present work belongs, focuses on harvesting ambient energy. In this category, continuous-time EH models are typically based on Poisson processes \cite{feng2021age,wu2017optimal,yates2021age,zheng2019closed,gindullina2021age,bacingolu2019optimal,elmagid2022age,arafa2019age,abdelmagid2022distribution}. In contrast, discrete-time EH models commonly rely on Bernoulli processes \cite{jia2021age,zhao2025age,hatami2022on,hatami2021aoi,chen2021optimization,delfani2025semantics,xu2023optimal,xiao2024infromation,ngo2025timely,jaiswal2023age,banerjee2024minimizing}, which is also adopted in this work.

In this paper, we employ energy as the enabling resource for executing actuations and integrate the semantics of the source and destination into the communication system design by focusing explicitly on the \emph{timeliness of actions}. Actions are commanded by data packets and enabled either by energy availability or by permissions from a controller. We show that although increasing the rate of status updates generally improves information freshness, it does not necessarily improve system performance when evaluated in terms of the timeliness of executed actions along the entire chain from data generation to reception and actuation. This observation highlights the limitations of AoI-centric evaluations and motivates a more action-aware, semantic perspective.

In our earlier work \cite{nikkhah2023age,nikkhah2023ageo}, we analyzed the timeliness of actions that occur immediately upon data reception and incorporated energy transmission into the system model. In contrast, the present paper considers systems that harvest ambient energy rather than relying on wireless power transfer. More importantly, we address scenarios in which a newly received data packet cannot be actuated immediately due to insufficient energy and must therefore be stored for future actuation. As a result of this storage capability, the age of the data triggering an action may vary at the time of actuation, implying that different actuations may differ in timeliness.

\subsection{Contributions}
The main contributions of this paper are:
\begin{itemize}
	\item We extend the \emph{Age of Actuation (AoA)} metric to characterize the timeliness of actions independently of the age of the data packets that trigger them.
	\item We introduce a new metric, termed the \emph{Age of Actuated Information (AoAI)}, which evaluates the timeliness of actions while explicitly accounting for the age of the actuating data packets.
	\item We analytically characterize both metrics across four scenarios within three cases of a unified system model and provide insights into their similarities and fundamental differences.
\end{itemize}

The two metrics coincide in instantaneous actuation systems, such as those studied in \cite{nikkhah2023age,nikkhah2023ageo}, but diverge in scenarios where data storage is possible, emphasizing the importance of accounting for the freshness of information at the moment of actuation.

\section{System Model} \label{SystemModel}
We consider four scenarios within three cases of a general system model, illustrated in Fig. \ref{fig:System Models}. The first case, depicted in Fig. \ref{fig:An Infinite Sized Queue with a Controller}, consists of a Geo/Geo/1 for storing data packets and a controller to give permissions to an actuator to actuate or not. This case is analyzed under two policies: First-Come First-Served (FCFS) and Last-Come First-Served (LCFS), with preemption in service, which are the first two scenarios. The third scenario, illustrated in Fig. \ref{fig:A Buffer with a Controller}, consists of a  Geo/Geo/1/1 instead of the Geo/Geo/1. Lastly, as a special case to illustrate what the controller can be in practice, we consider a battery unit, which would be the fourth scenario. The availability of energy in the battery represents the permission of the controller as an opportunity for an actuation. In all our discrete-time systems in Fig. \ref{fig:System Models}, time is divided into equal-length slots, indexed by $n$. The system includes an actuator at the receiver, which operates based on the information contained in data packets from a process occurring at the origin. A data packet is generated and successfully transmitted and arrives at the queue at each time slot with a probability of $\lambda_1$. In case the system fails to transmit the generated data to the queue successfully, the packet is dropped. We define the event $\Lambda_1(n)=\{0,1\}$ at time slot $
n$. Then, $\Lambda_1(n)=0$ denotes an unsuccessful, and $\Lambda_1(n)=1$ denotes a successful data packet reception at time $n$. 

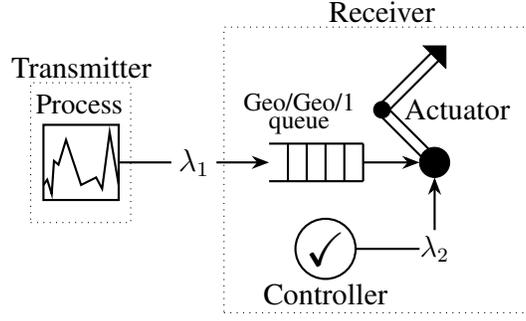
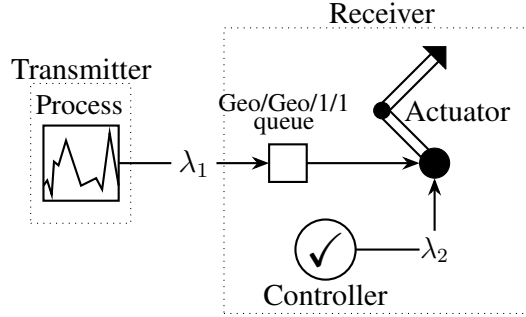
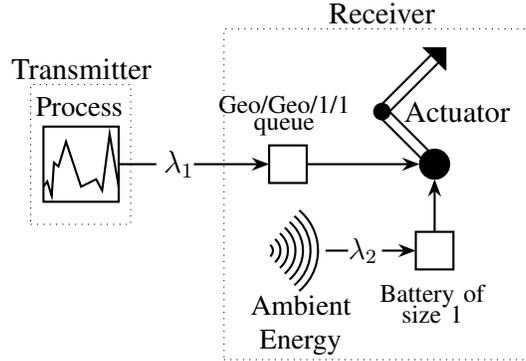
\begin{figure}[t]
\centering
\begin{subfigure}[b]{0.9\linewidth}
\centering
\begin{tikzpicture}

\draw [thick] (4.5,0) rectangle (5.5,1) node[scale=0.9] [xshift=-0.58cm,yshift=0.3cm] {Process}; 
\draw [thick] (4.5,0.2)-- (4.55,0.27) -- (4.61,0.08) --(4.64,0.52) -- (4.7,0.47) -- (4.8,0.8) -- (5,0.2) -- (5.2,0.3) -- (5.28,0.15) -- (5.38,0.9) -- (5.5,0.2);

\draw [ thick] (5.5,0.5) -- (6.2,0.5) node[xshift=0.3cm, yshift=0cm] {$\lambda_1$};

\draw[-Stealth, thick]  (6.8,0.5) -- (7.5,0.5) node [ xshift=0.4cm, yshift=0.5cm,scale=0.8] {queue} node [ xshift=0.4cm, yshift=0.8cm,scale=0.8] {Geo/Geo/1};;

\draw [thick] (7.5,0.25) -- (8.75,0.25);
\draw [thick] (7.5,0.75) -- (8.75,0.75);
\draw [thick] (7.75,0.25) -- (7.75,0.75);
\draw [thick] (8,0.25) -- (8,0.75);
\draw [thick] (8.25,0.25) -- (8.25,0.75);
\draw [thick] (8.5,0.25) -- (8.5,0.75);
\draw [thick] (8.75,0.25) -- (8.75,0.75);

\draw[-Stealth, thick] (9.5-0.75,1-0.5) -- (11-1.5,1-0.5);
\draw[fill,black] (11.2-1.5,1-0.5) circle (0.2cm);
\draw[thick] (11.25-1.5,1.05-0.5) -- (11.25-0.7-1.5,1.05+0.7-0.5);
\draw[thick] (11.15-1.5,0.95-0.5) -- (11.15-0.7-1.5,0.95+0.7-0.5);
\draw[fill,black] (9,1.70-0.5) circle (0.11cm);
\node [ xshift=10cm, yshift=1.20cm,scale=0.95] {Actuator};
\draw[thick] (10.55-1.5,1.65-0.5) -- (10.55+0.7-1.5,1.65+0.7-0.5);
\draw[thick] (10.45-1.5,1.75-0.5) -- (10.45+0.7-1.5,1.75+0.7-0.5);

\draw[fill=black] (10.55+0.8-1.5,1.65+0.6-0.5) -- (10.45+0.6-1.5,1.75+0.8-0.5) -- (10.55+0.8-1.5,1.75+0.8-0.5)-- (10.55+0.8-1.5,1.65+0.6-0.5);

\node [ xshift=8.25cm, yshift=-0.65cm,scale=1.5] {\checkmark} ;
\draw[thick]  (8.25,-0.65) circle (0.4) node [ xshift=0cm, yshift=-0.6cm,scale=0.95] {Controller} ;
\draw[ thick] (8.65,-0.65) -- (9.5,-0.65) node [ xshift=0.2cm, yshift=0cm,scale=0.95] {$\lambda_2$} ;

\draw[-Stealth, thick] (9.7,-0.4) -- (9.7,0.3);  
  
\draw [dotted] (4.33,-0.3) rectangle (5.66,1.56)  node [ xshift=-0.66cm, yshift=0.2cm,scale=0.95] {Transmitter};

\draw [dotted] (6.9,-1.5) rectangle (11,2.3) node [ xshift=-2cm, yshift=0.2cm,scale=0.95] {Receiver};

\end{tikzpicture}

\caption{A Geo/Geo/1 with a controller.}
\label{fig:An Infinite Sized Queue with a Controller}
\end{subfigure}

\vspace{1em} 

\begin{subfigure}[b]{0.9\linewidth}
\centering
\begin{tikzpicture}

\draw [thick] (4.5,0) rectangle (5.5,1) node[scale=0.9] [xshift=-0.58cm,yshift=0.3cm] {Process}; 
\draw [thick] (4.5,0.2)-- (4.55,0.27) -- (4.61,0.08) --(4.64,0.52) -- (4.7,0.47) -- (4.8,0.8) -- (5,0.2) -- (5.2,0.3) -- (5.28,0.15) -- (5.38,0.9) -- (5.5,0.2);

\draw [ thick] (5.5,0.5) -- (6.2,0.5) node[xshift=0.3cm, yshift=0cm] {$\lambda_1$};

\draw[-Stealth, thick]  (6.8,0.5) -- (7.5,0.5);
\draw [thick] (7.5,0.75) rectangle (8,0.25) node [ xshift=-0.3cm, yshift=0.75cm,scale=0.8] {queue} node [ xshift=-0.3cm, yshift=1.05cm,scale=0.8] {Geo/Geo/1/1};

\draw[-Stealth, thick] (9.5-1.5,1-0.5) -- (11-1.5,1-0.5);
\draw[fill,black] (11.2-1.5,1-0.5) circle (0.2cm);
\draw[thick] (11.25-1.5,1.05-0.5) -- (11.25-0.7-1.5,1.05+0.7-0.5);
\draw[thick] (11.15-1.5,0.95-0.5) -- (11.15-0.7-1.5,0.95+0.7-0.5);
\draw[fill,black] (9,1.70-0.5) circle (0.11cm);
\node [ xshift=10cm, yshift=1.20cm,scale=0.95] {Actuator};
\draw[thick] (10.55-1.5,1.65-0.5) -- (10.55+0.7-1.5,1.65+0.7-0.5);
\draw[thick] (10.45-1.5,1.75-0.5) -- (10.45+0.7-1.5,1.75+0.7-0.5);

\draw[fill=black] (10.55+0.8-1.5,1.65+0.6-0.5) -- (10.45+0.6-1.5,1.75+0.8-0.5) -- (10.55+0.8-1.5,1.75+0.8-0.5)-- (10.55+0.8-1.5,1.65+0.6-0.5);

\node [ xshift=8.25cm, yshift=-0.65cm,scale=1.5] {\checkmark} ;
\draw[thick]  (8.25,-0.65) circle (0.4) node [ xshift=0cm, yshift=-0.6cm,scale=0.95] {Controller} ;
\draw[ thick] (8.65,-0.65) -- (9.5,-0.65) node [ xshift=0.2cm, yshift=0cm,scale=0.95] {$\lambda_2$} ;

\draw[-Stealth, thick] (9.7,-0.4) -- (9.7,0.3);

\draw [dotted] (4.33,-0.3) rectangle (5.66,1.56)  node [ xshift=-0.66cm, yshift=0.2cm,scale=0.95] {Transmitter};

\draw [dotted] (6.9,-1.5) rectangle (11,2.3) node [ xshift=-2cm, yshift=0.2cm,scale=0.95] {Receiver};

\end{tikzpicture}

\caption{A Geo/Geo/1/1 with a controller.}
\label{fig:A Buffer with a Controller}
\end{subfigure}

\vspace{1em}

\begin{subfigure}[b]{0.9\linewidth}
\centering
\begin{tikzpicture}

\draw [thick] (4.5,0) rectangle (5.5,1) node[scale=0.9] [xshift=-0.58cm,yshift=0.3cm] {Process}; 
\draw [thick] (4.5,0.2)-- (4.55,0.27) -- (4.61,0.08) --(4.64,0.52) -- (4.7,0.47) -- (4.8,0.8) -- (5,0.2) -- (5.2,0.3) -- (5.28,0.15) -- (5.38,0.9) -- (5.5,0.2);

\draw [thick] (5.5,0.5) -- (6,0.5) node[xshift=0.3cm, yshift=0cm] {$\lambda_1$};

 \foreach \r in {0.1,0.2,...,0.7} 
    \draw[thick,shift={(7.45cm,-0.65cm)}] (0,0) ++(-50:\r cm) arc (-50:50:\r cm) ;
    \draw[thick] (8.25,-0.65) -- (8.55,-0.65) node[scale=0.9,align=center, xshift=-0.7cm,yshift=-1.1cm]  {Ambient \\ Energy} ;
    \draw[-Stealth, thick] (9,-0.65) -- (9.7-0.25,-0.65) node [ xshift=-0.7cm, yshift=0cm,scale=0.95] {$\lambda_2$};

\draw[-Stealth, thick]  (6.5,0.5) -- (7.5,0.5);

\draw [thick] (7.5,0.75) rectangle (8,0.25) node [ xshift=-0.3cm, yshift=0.75cm,scale=0.8] {queue} node [ xshift=-0.3cm, yshift=1.05cm,scale=0.8] {Geo/Geo/1/1};

\draw[-Stealth, thick] (9.5-1.5,1-0.5) -- (11-1.5,1-0.5);
\draw[fill,black] (11.2-1.5,1-0.5) circle (0.2cm);
\draw[thick] (11.25-1.5,1.05-0.5) -- (11.25-0.7-1.5,1.05+0.7-0.5);
\draw[thick] (11.15-1.5,0.95-0.5) -- (11.15-0.7-1.5,0.95+0.7-0.5);
\draw[fill,black] (9,1.70-0.5) circle (0.11cm);
\node [ xshift=10cm, yshift=1.20cm,scale=0.95] {Actuator};
\draw[thick] (10.55-1.5,1.65-0.5) -- (10.55+0.7-1.5,1.65+0.7-0.5);
\draw[thick] (10.45-1.5,1.75-0.5) -- (10.45+0.7-1.5,1.75+0.7-0.5);

\draw[fill=black] (10.55+0.8-1.5,1.65+0.6-0.5) -- (10.45+0.6-1.5,1.75+0.8-0.5) -- (10.55+0.8-1.5,1.75+0.8-0.5)-- (10.55+0.8-1.5,1.65+0.6-0.5);

\draw[thick] (9.7-0.25,-0.4) rectangle (9.7+0.25,-0.9) node [ xshift=-0.28cm, yshift=-0.35cm,scale=0.8] {Battery of} node [ xshift=-0.28cm, yshift=-0.6cm,scale=0.8] {size 1};

\draw[-Stealth, thick] (9.7,-0.4) -- (9.7,0.3);

\draw [dotted] (4.33,-0.3) rectangle (5.66,1.56)  node [ xshift=-0.66cm, yshift=0.2cm,scale=0.95] {Transmitter};

\draw [dotted] (6.9,-2.1) rectangle (11,2.3) node [ xshift=-2cm, yshift=0.2cm,scale=0.95] {Receiver};

\end{tikzpicture}

\caption{A Geo/Geo/1/1 with a one-sized battery}
\label{fig:A Buffer with a One-Sized Battery}
\end{subfigure}

\caption{The considered system model in three cases that capture different aspects.}
\label{fig:System Models}
\end{figure}

\subsection{A Geo/Geo/1 with a Controller} \label{An Infinite-sized Queue with a Controller}
Consider the case depicted in Fig. \ref{fig:An Infinite Sized Queue with a Controller}. A Geo/Geo/1 is available to store data packets for later actuations. With the probability of $\lambda_2$ at each time slot, the controller permits an actuation. We define the event $\Lambda_2(n)=\{0,1\}$, where $\Lambda_2(n)=0$ denotes no permission, and $\Lambda_2(n)=1$ denotes permission for an actuation.
The state of the queue is represented by $Q(n) \in \{0,1,2,\hdots \}$, where it indicates the number of data packets in the queue at the end of time slot $n$. The dynamics of the queue can be described by a process of Geo/Geo/1 with the arrival rate of $\lambda_1$ and the service rate of $\lambda_2$.

At the beginning of each time slot, the system checks for the controller's permission and the availability of a data packet, whether stored or newly received in that time slot. If both are checked, then the actuator actuates the oldest packet under the FCFS policy and the freshest packet under the LCFS policy with pre-emption, is service\footnote{A newly arrived packet is prioritized for actuation even within the arrival time slot.}.  In other words
\begin{equation} \label{Actuation_Equation}
\alpha(n)=
\begin{cases}
1 &  \Lambda_2(n)=1 \land (Q(n-1)\neq 0 \lor \Lambda_1(n)=1),\\
0 &  \text{otherwise},
\end{cases}
\end{equation}
where in time slot $n$, $\alpha(n)=1$ indicates an actuation while $\alpha(n)=0$ indicates no actuation.

\subsection{A Geo/Geo/1/1 with a Controller} \label{A Geo/Geo/1/1 with a Controller}

Consider the case depicted in Fig. \ref{fig:A Buffer with a Controller}. Everything is the same as section \ref{An Infinite-sized Queue with a Controller} except there is a Geo/Geo/1/1 to store data packets. This represents the need for only the freshest data packet at each time slot. 
Older data packets can be discarded once the most recent data packet is stored. A data packet remains in the Geo/Geo/1/1 if it is the freshest available and has not been utilized for actuation. Therefore, upon receiving a new data packet, it replaces the existing one in the queue. Also, once a data packet is actuated, it is removed from the queue. Expression (\ref{Actuation_Equation}) also holds for this scenario.

We can represent the state of the system using a Markov chain, $Q(n) \in \{0,1\}$, to capture the dynamics of the system, where $Q(n)=0$ signifies an empty queue and $Q(n)=1$ represents a full queue at the end of the time slot $n$. The transition probability matrix for this Markov chain is
\setcounter{MaxMatrixCols}{30}
\begin{equation*} \label{TPM_SystemModel_Geo/Geo/1/1_Controller}
\mathbf{P}_s=
\begin{bmatrix}
\bar{\lambda}_1  + \lambda_1 \lambda_2 &  \lambda_1 \bar{\lambda}_2 \\
\lambda_2 & \bar{\lambda}_2  \\
\end{bmatrix}.
\end{equation*}

The first row and the first column represent the state $Q(n)=0$ and the second row and the second column represent the state $Q(n)=1$.

\subsection{A Geo/Geo/1/1 with a Battery} \label{Battery representing the controller}

As a practical instance of a controller, consider a battery instead, depicted in Fig. \ref{fig:A Buffer with a One-Sized Battery}. The battery is capable of storing a single energy unit. With a probability of $\lambda_2$ at each time slot, ambient energy (such as solar, wind, vibration, and radio frequency) can be harvested to charge the battery with one energy unit. Here, $\Lambda_2(n)=0$ denotes unavailability and $\Lambda_2(n)=1$ denotes the availability of an energy unit to harvest.

A two-dimensional Markov chain, denoted as $(Q(n), B(n))$, captures the dynamics of the system states. The state of the queue is represented the same way as section \ref{A Geo/Geo/1/1 with a Controller}. Similarly, $B(n) \in \{0,1\}$ indicates the battery's state, where $B(n)=0$ represents an empty battery and $B(n)=1$ represents a full battery at the end of time slot $n$. The system can be in three states: $(0,0)$, $(0,1)$, and $(1,0)$. Note that state $(1,1)$ is not feasible, as if the queue contains an unused data packet and sufficient energy is available, both would be utilized concurrently for actuation. Consequently, the transition probability matrix for this Markov chain is
\setcounter{MaxMatrixCols}{30}
\begin{equation*} \label{TPM_SystemModel_Geo/Geo/1/1_Battery}
\mathbf{P}_s=
\begin{bmatrix}
\lambda_1 \lambda_2 + \bar{\lambda}_1 \bar{\lambda}_2 &  \bar{\lambda}_1 \lambda_2 & \lambda_1 \bar{\lambda}_2\\
\lambda_1 \bar{\lambda}_2 & \lambda_2 +\bar{\lambda}_1 \bar{\lambda}_2 &0 \\
 \lambda_2  & 0 & \bar{\lambda}_2  \\
\end{bmatrix}.
\end{equation*}

The rows and columns represent the joint states $(Q(n),B(n))$ for states $(0,0)$, $(0,1)$, and $(1,0)$, respectively.

\section{Analytical Results} 
This section presents the definitions and analysis for the AoA and the AoAI for the considered scenarios and the cases of the system model. We will show the differences and identify potential similarities between AoI, AoA, and AoAI.

\subsection{Age of Information (AoI)}
AoI represents the time elapsed since the generation of the freshest data packet available. It is defined as $I(t)=t-u(t)$, where $t$ is the time instant and $u(t)$ is the timestamp of the freshest successfully received data packet. Consequently, in our system model, in each time slot $n$, the AoI is reset to $1$ after a successful reception of a data packet (with probability $\lambda_1$), and increases by $1$ in the event of a failure in reception of a data packet (with probability $1-\lambda_1$):
\begin{equation}
I(n)=
\begin{cases}
1 & \Lambda_1(n)=1,\\
I(n-1)+1 & \Lambda_1(n)=0.
\end{cases}
\end{equation}

Then, the average AoI would be $\bar{I}=\frac{1}{\lambda_1}$ \cite{pappas2023age} for all the cases and scenarios discussed in section \ref{SystemModel}.

\subsection{Age of Actuation (AoA)} \label{AoASection}
AoA is defined as $A(t)=t-a(t)$, where $a(t)$ is the timestamp of the last performed actuation. In our system model, at each time slot $n$, AoA is reset to $1$ if there is an actuation and increases by $1$ otherwise, irrespective of the age of the actuated data packet:
\begin{equation}
A(n)=
\begin{cases}
1 &  \alpha(n)=1,\\
A(n-1)+1 & \alpha(n)=0.
\end{cases}
\end{equation}

\textit{This metric is relevant for capturing the timeliness of actions without considering the freshness of the actuated information.}

\begin{theorem} \label{theo:AoA}
The closed-form expressions of the average AoA for the four scenarios defined in the three cases of the system model are as follows:

\begin{enumerate}[label=(\roman*)]

\item \label{theo:The Average AoA Infinite} The average AoA for the Geo/Geo/1 with a controller under both policies of FCFS and LCFS, is 
\begin{equation} \label{TheAverageAoAforQueue}
\bar{A}=\frac{1}{\lambda_1}.
\end{equation}

\item  \label{theo:The Average AoA cache} The average AoA for the Geo/Geo/1/1 with a controller is
\begin{equation}
\bar{A}=\frac{ (\lambda_2-1) (\lambda_1^2+ \lambda_1 \lambda_2)   - \lambda_2^2}{\lambda_1 \lambda_2 \left(\lambda_1  ( \lambda_2-1) - \lambda_2\right) }.
\end{equation}

\item  \label{theo:The Average AoA battery} The average AoA for the Geo/Geo/1/1 with a battery is given by 
\begin{equation} \label{Average_AoA_Battery}
\bar{A}=
\frac{\lambda_{1}^{4} \left(\lambda_{2} - 1\right)^{3} - 2 \lambda_{1}^{3} \lambda_{2} \left(\lambda_{2} - 1\right)^{2} - \lambda_{1}^{2} \lambda_{2}^{2} \left(2 \lambda_{2}^{2} - 3 \lambda_{2} + 1\right) + \lambda_{1} \lambda_{2}^{3} \left(3 \lambda_{2} - 2\right) - \lambda_{2}^{4}}{\lambda_{1} \lambda_{2} \left(\lambda_{1} \left(\lambda_{2} - 1\right) - \lambda_{2}\right) \left(\lambda_{1}^{2} \left(\lambda_{2} - 1\right)^{2} + \lambda_{1} \left(- 2 \lambda_{2}^{2} + \lambda_{2}\right) + \lambda_{2}^{2}\right)}.
\end{equation}

\end{enumerate}

\end{theorem}

\begin{proof}
See Appendices \ref{Proof of The Average AoA Infinite}, \ref{Proof of The Average AoA cache}, and \ref{Proof of The Average AoA battery}.    
\end{proof}

\begin{remark} \label{TheIntuituionForIAequality}
In scenarios Geo/Geo/1 with a controller under both policies of FCFS and LCFS, the intuition behind the independence of the average AoA from $\lambda_2$ and its equality to the average AoI is that each successfully received data packet resets the AoI to $1$ \textit{instantaneously} and resets AoA to $1$ \textit{ultimately}. Although, for the AoI to reset to $1$ only a received data packet is needed and for the AoA to reset to $1$ a permission is also required, with the assumption of stability, i.e. $\lambda_1 < \lambda_2$, there will be enough permissions from the controller to actuate all the data packets.
However, if the stability condition is violated, $\lambda_1 \geq \lambda_2$, we have $\bar{A}=\frac{1}{\lambda_2} $. Thus, generally $\bar{A}=\max\{\frac{1}{\lambda_1},\frac{1}{\lambda_2}\}$.
\end{remark}

\begin{remark}
In scenarios Geo/Geo/1 with a controller under both policies of FCFS and LCFS, the average AoA is the same for both policies since the age of the packet does not affect the AoA, and what matters is only the number of packets in the queue. However, this would not be valid for the average AoAI.
\end{remark}

\subsection{Age of Actuated Information (AoAI)}
AoAI measures the time elapsed since the generation of the freshest actuated data packet. We generally define it as $AI(t)=t-a_{AI}(t)+I_{p}(a_{AI}(t))$, where $a_{AI}(t)$ is the timestamp of the last actuation that has reset the AoAI (AoAI-resetting) and $I_{p}(a_{AI}(t))$ is the age of the actuated packet at the last AoAI-resetting actuation. In our system model, at each time slot $n$, upon an actuation, if the age of the actuated packet is lower\footnote{It is either greater or lower and cannot be equal since there cannot be two packets generated at the same time slot.} than the current AoAI\footnote{The only case that it does not generally happen is the scenario of Geo/Geo/1 under LCFS policy, since packets that become actuated might have been generated before the already actuated packets.}, the AoAI is reset to the age of the new actuated packet and increases by $1$ otherwise:
\begin{equation} \label{AoAI_Evoultion}
AI(n)=
\begin{cases}
\min\{AI(n-1)+1,I_p(n)\} & \alpha(n)=1,\\
AI(n-1)+1 & \alpha(n)=0,
\end{cases}
\end{equation}
where $I_{p}(n)$ is the age of the actuated packet at time slot $n$. 
\textit{This metric assesses the timeliness of actions in relation to the freshness of the actuated data packets.}

\begin{theorem} \label{theo:AoAI}
The closed-form expressions of the average AoAI for the four scenarios defined in the three cases of the system model are as follows:

\begin{enumerate}[label=(\roman*)]

\item \label{theo:The Average AoAI Infinite} The average AoAI for the Geo/Geo/1 \textit{under the policy of FCFS} is 
\begin{equation} \label{AoAI_average_InfCache}
\overline{AI}=\frac{\lambda_1^2 (\lambda_2-1) (\lambda_1 - \lambda_2) + \lambda_2^3( \lambda_1-1)}{\lambda_1 (\lambda_1 - \lambda_2) \lambda_2^2}.
\end{equation}

\item  \label{theo:The Average AoAI cache} The average AoAI for the Geo/Geo/1 with a controller \textit{under the policy of LCFC} and Geo/Geo/1/1 with a controller, both with pre-emption in service, is given by
\begin{equation}
\overline{AI}=\frac{1}{\lambda_1}+\frac{1}{\lambda_2}-1.
\end{equation}

\item  \label{theo:The Average AoAI battery} The average AoAI for the Geo/Geo/1/1 with a battery is given by  
{\footnotesize
\begin{equation} \label{Average_AoAI_Battery}
\overline{AI}=\frac{\lambda_1^4 \left(\lambda_2-4\right) \left(\lambda_2-1\right)^3 \lambda_2 - 4 \lambda_1^3 \left( \lambda_2-1\right)^3 \lambda_2^2 + \lambda_1 \left(3 - 4 \lambda_2\right) \lambda_2^4 + \lambda_2^5 + \lambda_1^5 \left(\lambda_2-1\right)^3 \left(2 \lambda_2-1\right) + 2 \lambda_1^2 \lambda_2^3 \left(2 - 5 \lambda_2 + 3 \lambda_2^2\right)}{\lambda_1 \lambda_2 \left(\lambda_1 + \lambda_2 - \lambda_1 \lambda_2\right)^2 \left(\lambda_1^2 \left(\lambda_2-1\right)^2 + \lambda_2^2 + \lambda_1 \left(\lambda_2 - 2 \lambda_2^2\right)\right)}.
\vspace{-10pt}
\end{equation}}

\end{enumerate}

\end{theorem}

\begin{proof}
See Appendices \ref{Proof of The Average AoAI Infinite}, \ref{Proof of The Average AoAI cache}, and \ref{Proof of The Average AoAI battery}.
\end{proof}

\begin{remark}
The AoAI is the same for the Geo/Geo/1 with a controller \textit{under the policy of LCFS} and the Geo/Geo/1/1 with a controller, both with preemption in service. This is because the AoAI in the LCFS scenario is reset only when a fresher data packet has been actuated, which is equivalent to having a single-size queue to store at most one freshest, not-yet-actuated packet at any time. This is only valid for AoAI, and these two scenarios differ in AoA, as seen in section \ref{AoASection}.
\end{remark}
\begin{remark}
For the Geo/Geo/1/1 with a controller, we can specify the general definition of the AoAI to $AI(t)=t-a(t)+I(a(t))=A(t)+I(a(t))$, where $I(a(t))$ represents the AoI at the time of the last actuation. Consequently, in our system model, for the Geo/Geo/1/1 with a controller, AoAI is reset to the current AoI upon an actuation and increases by $1$ otherwise. Thus, for this scenario, (\ref{AoAI_Evoultion}) is specified to 
\begin{equation}
AI(n)=
\begin{cases}
I(n) & \alpha(n)=1,\\
AI(n-1)+1 & \alpha(n)=0.
\end{cases}
\end{equation}
This is because, at any time, the actuation utilizes the freshest data packet available.
\end{remark}

\noindent
The results from Theorems \ref{theo:AoA} and \ref{theo:AoAI} are summarized in Table~\ref{tableofFormulas}.

\begin{table}
\centering
\caption{Average AoI, Average AoA, and Average AoAI for different scenarios.}

\label{tableofFormulas}
\begin{tabular}{c|
    >{\centering\arraybackslash}p{3.8cm}|
    >{\centering\arraybackslash}p{3.1cm}|
    >{\centering\arraybackslash}p{3.4cm}|
    >{\centering\arraybackslash}p{2.9cm}|}
\cline{2-5}
& Geo/Geo/1 \newline with a controller \newline under FCFS & Geo/Geo/1 \newline with a controller \newline under LCFS  & Geo/Geo/1/1 \newline with a controller & Geo/Geo/1/1 \newline with a battery   \\ \hline
\multicolumn{1}{|c|}{Average AoI}  &$\frac{1}{\lambda_1}$ &$\frac{1}{\lambda_1}$ &$\frac{1}{\lambda_1}$ &$\frac{1}{\lambda_1}$  \\ \hline
\multicolumn{1}{|c|}{Average AoA} &$\frac{1}{\lambda_1}$ & $\frac{1}{\lambda_1}$ &$\frac{ (\lambda_2-1) (\lambda_1^2+ \lambda_1 \lambda_2)   - \lambda_2^2}{\lambda_1 \lambda_2 \left(\lambda_1  ( \lambda_2-1) - \lambda_2\right) }$ &  (\ref{Average_AoA_Battery}) \\ \hline
\multicolumn{1}{|c|}{Average AoAI} &   $\frac{\lambda_1^2 (\lambda_2-1) (\lambda_1 - \lambda_2) + \lambda_2^3( \lambda_1-1)}{\lambda_1 (\lambda_1 - \lambda_2) \lambda_2^2}$ &$\frac{1}{\lambda_1}+\frac{1}{\lambda_2}-1$ &$\frac{1}{\lambda_1}+\frac{1}{\lambda_2}-1$ &(\ref{Average_AoAI_Battery}) \\ \hline
\end{tabular}
\end{table}

Through the stages of modeling and solving the Markov chain to reach (\ref{AoAI_average_InfCache}), we could obtain an insightful, novel, and fundamental result on queuing timeliness. It is presented in the following theorem.

\begin{definition}[Queue Markov Process and Queue State]\label{def:Queue_state}
We define a Markov process for the queue condition as $C(r)~=~\{\Phi_{r-1},\Phi_{r-2},\cdots,\Phi_{2},\Phi_{1}\}$ where $\Phi_i=0$ and $\Phi_i=1$ indicate the availability and unavailability of the packet with age $i$ in the queue, respectively, where $i \in \{1,\cdots,r-1\}$. We define $h=\max\{ i \in \{1,\cdots,r-1\}: \Phi_{i}=1\}$. Thus, we have $C(r)~=~\{\Phi_{r-1},\Phi_{r-2},\cdots,\Phi_{h+1},\Phi_{h},\Phi_{h-1},\cdots,\Phi_{2},\Phi_{1}\}$, where $\Phi_{h}=1$ and $\Phi_{i}=0 \  \forall i \in \{h+1,h+2,\cdots,r-1\}$. We have $r  \in \mathbb{N} , h \in \mathbb{N}_0
, r > h$. If $h \in \mathbb{N}$, then $h$ indicates the age of the oldest packet in a non-empty queue. If $h=0$, then $C$ is an empty queue. Also, $l=\sum_{i=1}^{h}\Phi_{i}$ is the length of the queue. 
We call $C$ a \textit{queue~Markov~process} and each distinct subset of $C$ with an arbitrary $h$ a \textit{queue~state}.
\end{definition}

\begin{theorem} \label{theo:Closed_form}
Consider a Geo/Geo/1, with arrival probability of $\lambda_1$ and service probability of $\lambda_2$, under the FCFS policy.
The stationary probability for this queue Markov process is only dependent on the length of the queue, i.e., $l$, and the age of the oldest packet in the queue (head of the queue since the policy is FCFS), i.e., $h$, and is described by 
\begin{equation} \label{Final_Closed_Form}
\Gamma(h,l)=x^{l-1}z^{h-l}\Gamma_1+\sum_{i=1}^{l-1} \sum_{j=0}^{h-l} {l-1 \choose i} {h-l \choose j} w^{i} x^{l-i-1} y^{j} z^{h-l-j} f_{j+i}  + \sum_{j=1}^{h-l} {h-l \choose j} x^{l-1} y^{j} z^{h-l-j} f_{j}, 
\end{equation}
where $w=\lambda_1 \lambda_2, \ x=\lambda_1 \bar{\lambda}_2, \ y=\bar{\lambda}_1 \lambda_2, \ z=\bar{\lambda}_1 \bar{\lambda}_2 $, $D_{i}=(\frac{x}{y})^{i}(1-\frac{x}{y}) \ i \in \{0,1,2,\cdots\}$, $f_i=xD_i+wD_{i+1} \ i \in \{1,2,\cdots\}$, and $\Gamma_{1}=(\frac{1}{1-w})(xD_0+wzD_1+wyD_2)$.
\end{theorem}
\begin{proof}
See Appendix \ref{Proof of The Average AoAI Infinite}.
\end{proof}
\begin{corollary}
In a Geo/Geo/1, under the FCFS policy, for a given value of the age of the oldest packet, the stationary distribution of the queue Markov process (each queue state) only depends on the length of the queue and vice versa.
\end{corollary}

\begin{figure}[h!]
\centering 
\begin{subfigure}[b]{0.9\linewidth}
    \centering
\scalebox{0.8}{ \boldmath{
\begin{tikzpicture}

{
    \pgfdeclarepatternformonly{Sparse Vertical Lines}
    {%
        \pgfqpoint{-1pt}{-1pt}%
    }
    {%
        \pgfqpoint{10pt}{10pt}%
    }
    {%
        \pgfqpoint{9pt}{9pt}%
    }
    {
        \pgfsetlinewidth{0.1pt} 
        \pgfpathmoveto{\pgfqpoint{0pt}{0pt}}
        \pgfpathlineto{\pgfqpoint{0pt}{9.1pt}}
        \pgfusepath{stroke}
    }
}

{
    \pgfdeclarepatternformonly{Sparse Horizontal Lines}
    {%
        \pgfqpoint{-1pt}{-1pt}%
    }
    {%
        \pgfqpoint{10pt}{10pt}%
    }
    {%
        \pgfqpoint{9pt}{9pt}%
    }
    {
        \pgfsetlinewidth{0.1pt} 
        \pgfpathmoveto{\pgfqpoint{0pt}{0pt}}
        \pgfpathlineto{\pgfqpoint{9.1pt}{0pt}}
        \pgfusepath{stroke}
    }
    }

{
    \pgfdeclarepatternformonly{Sparse North East Lines}
    {%
        \pgfqpoint{-1pt}{-1pt}%
    }
    {%
        \pgfqpoint{10pt}{10pt}%
    }
    {%
        \pgfqpoint{9pt}{9pt}%
    }
    {
        \pgfsetlinewidth{0.1pt} 
        \pgfpathmoveto{\pgfqpoint{0pt}{0pt}}
        \pgfpathlineto{\pgfqpoint{9.1pt}{9.1pt}}
        \pgfusepath{stroke}
    }
    }

    \pgfdeclarepatternformonly{Vertical Lines}
    {%
        \pgfqpoint{-1pt}{-1pt}%
    }
    {%
        \pgfqpoint{4pt}{4pt}%
    }
    {%
        \pgfqpoint{3pt}{3pt}%
    }
    {
        \pgfsetlinewidth{0.1pt} 
        \pgfpathmoveto{\pgfqpoint{0pt}{0pt}}
        \pgfpathlineto{\pgfqpoint{0pt}{3.1pt}}
        \pgfusepath{stroke}
    }

    \pgfdeclarepatternformonly{Horizontal Lines}
    {%
        \pgfqpoint{-1pt}{-1pt}%
    }
    {%
        \pgfqpoint{4pt}{4pt}%
    }
    {%
        \pgfqpoint{3pt}{3pt}%
    }
    {
        \pgfsetlinewidth{0.1pt} 
        \pgfpathmoveto{\pgfqpoint{0pt}{0pt}}
        \pgfpathlineto{\pgfqpoint{3.1pt}{0pt}}
        \pgfusepath{stroke}
    }

    \pgfdeclarepatternformonly{North East Lines}
    {%
        \pgfqpoint{-1pt}{-1pt}%
    }
    {%
        \pgfqpoint{4pt}{4pt}%
    }
    {%
        \pgfqpoint{3pt}{3pt}%
    }
    {
        \pgfsetlinewidth{0.1pt} 
        \pgfpathmoveto{\pgfqpoint{0pt}{0pt}}
        \pgfpathlineto{\pgfqpoint{3.1pt}{3.1pt}}
        \pgfusepath{stroke}
    }

\fill[pattern=Sparse Vertical Lines] (0,0) rectangle (1,1);
\fill[pattern=Sparse Vertical Lines] (1,0) rectangle (2,2);
\fill[pattern=Sparse Vertical Lines] (2,0) rectangle (3,1);
\fill[pattern=Sparse Vertical Lines] (3,0) rectangle (4,2);
\fill[pattern=Sparse Vertical Lines] (4,0) rectangle (5,1);
\fill[pattern=Sparse Vertical Lines] (5,0) rectangle (6,2);
\fill[pattern=Sparse Vertical Lines] (6,0) rectangle (7,3);
\fill[pattern=Sparse Vertical Lines] (7,0) rectangle (8,4);

\draw[line width=4pt, green]  (0,1) -- (1,1) -- (1,2) -- (2,2) -- (2,1) -- (3,1) -- (3,2) -- (4,2) -- (4,1) -- (5,1) -- (5,2) -- (6,2) -- (6,3) -- (7,3) -- (7,4) -- (8,4);

\fill[pattern=Sparse Horizontal Lines] (0,0) rectangle (1,1);
\fill[pattern=Sparse Horizontal Lines] (1,0) rectangle (2,2);
\fill[pattern=Sparse Horizontal Lines] (2,0) rectangle (3,3);
\fill[pattern=Sparse Horizontal Lines] (3,0) rectangle (4,4);
\fill[pattern=Sparse Horizontal Lines] (4,0) rectangle (5,5);
\fill[pattern=Sparse Horizontal Lines] (5,0) rectangle (6,1);
\fill[pattern=Sparse Horizontal Lines] (6,0) rectangle (7,2);
\fill[pattern=Sparse Horizontal Lines] (7,0) rectangle (8,1);

\draw[line width=2pt, blue]  (0,1) -- (1,1) -- (1,2) -- (2,2) -- (2,3) -- (3,3) -- (3,4) -- (4,4) -- (4,5) -- (5,5) -- (5,1) -- (6,1) -- (6,2) -- (7,2) -- (7,1) -- (8,1) ;

\fill[pattern=Sparse North East Lines] (0,0) rectangle (1,1);
\fill[pattern=Sparse North East Lines] (1,0) rectangle (2,2);
\fill[pattern=Sparse North East Lines] (2,0) rectangle (3,3);
\fill[pattern=Sparse North East Lines] (3,0) rectangle (4,4);
\fill[pattern=Sparse North East Lines] (4,0) rectangle (5,5);
\fill[pattern=Sparse North East Lines] (5,0) rectangle (6,4);
\fill[pattern=Sparse North East Lines] (6,0) rectangle (7,5);
\fill[pattern=Sparse North East Lines] (7,0) rectangle (8,4);

\draw[line width=0.8pt, red] (0,1) -- (1,1) -- (1,2) -- (2,2) -- (2,3) -- (3,3) -- (3,4) -- (4,4) -- (4,5) -- (5,5) -- (5,4) -- (6,4) -- (6,5) -- (7,5) -- (7,4) -- (8,4);

\draw[-Stealth, very thick ] (0,0) -- (9,0)  node [xshift=-0.2cm, yshift=0.4cm] {$n$} node [xshift=-9cm, yshift=-0.4cm] {$0$}  node [xshift=-8cm, yshift=-0.4cm] {$1$}  node [xshift=-7cm, yshift=-0.4cm] {$2$}  node [xshift=-6cm, yshift=-0.4cm] {$3$}  node [xshift=-5cm, yshift=-0.4cm] {$4$}  node [xshift=-4cm, yshift=-0.4cm] {$5$}  node [xshift=-3cm, yshift=-0.4cm] {$6$}  node [xshift=-2cm, yshift=-0.4cm] {$7$}  node [xshift=-1cm, yshift=-0.4cm] {$8$}  ;

\draw[-Stealth, very thick ] (0,0) -- (0,6) node [xshift=0.7cm, yshift=-0.2cm] {$AI(n)$}  node [xshift=0.7cm, yshift=-0.7cm] {$A(n)$} node [xshift=0.7cm, yshift=-1.2cm] {$I(n)$} node [xshift=-0.3cm, yshift=-5cm] {$1$} node [xshift=-0.3cm, yshift=-6cm] {$0$};

\draw (1,-0.1) -- (1,0.1);
\draw (2,-0.1) -- (2,0.1);
\draw (3,-0.1) -- (3,0.1);
\draw (4,-0.1) -- (4,0.1);
\draw (5,-0.1) -- (5,0.1);
\draw (6,-0.1) -- (6,0.1);
\draw (7,-0.1) -- (7,0.1);
\draw (8,-0.1) -- (8,0.1);

\draw (-0.1,1) -- (0.1,1);
\draw (-0.1,2) -- (0.1,2);
\draw (-0.1,3) -- (0.1,3);
\draw (-0.1,4) -- (0.1,4);
\draw (-0.1,5) -- (0.1,5);

\node[xshift=2cm,yshift=5.6cm] {AoAI};
\node[xshift=2cm,yshift=5.1cm] {AoA};
\node[xshift=2cm,yshift=4.6cm] {AoI};

\fill[pattern=North East Lines] (3.5-1,6.3-0.5-0.4) rectangle (4.5-1,6.2-0.4);
\fill[pattern=Horizontal Lines] (3.5-1,5.3-0.4) rectangle (4.5-1,6.2-0.5-0.4);
\fill[pattern=Vertical Lines] (3.5-1,5.3-0.5-0.4) rectangle (4.5-1,5.2-0.4);

\draw[line width = 2pt, red] (3.5-1,6.2-0.4) -- (4.5-1,6.2-0.4);
\draw[line width = 2pt, blue]  (3.5-1,6.2-0.5-0.4) rectangle (4.5-1,6.2-0.5-0.4);
\draw[line width = 2pt, green] (3.5-1,5.2-0.4) rectangle (4.5-1,5.2-0.4);

\draw (1.5,4.7-0.4) rectangle (3.7,6.3-0.4);

\end{tikzpicture}}}
\caption{Geo/Geo/1 under FCFS policy with a controller.}
\label{fig:SamplePathInfiniteControllerFCFS}
\vspace{10pt}
\end{subfigure}

\begin{subfigure}[b]{0.9\linewidth}
    \centering
\scalebox{0.8}{ \boldmath{
\begin{tikzpicture}

\fill[pattern=Sparse Vertical Lines] (0,0) rectangle (1,1);
\fill[pattern=Sparse Vertical Lines] (1,0) rectangle (2,2);
\fill[pattern=Sparse Vertical Lines] (2,0) rectangle (3,1);
\fill[pattern=Sparse Vertical Lines] (3,0) rectangle (4,2);
\fill[pattern=Sparse Vertical Lines] (4,0) rectangle (5,1);
\fill[pattern=Sparse Vertical Lines] (5,0) rectangle (6,2);
\fill[pattern=Sparse Vertical Lines] (6,0) rectangle (7,3);
\fill[pattern=Sparse Vertical Lines] (7,0) rectangle (8,4);

\draw[line width=4pt, green]  (0,1) -- (1,1) -- (1,2) -- (2,2) -- (2,1) -- (3,1) -- (3,2) -- (4,2) -- (4,1) -- (5,1) -- (5,2) -- (6,2) -- (6,3) -- (7,3) -- (7,4) -- (8,4);

\fill[pattern=Sparse Horizontal Lines] (0,0) rectangle (1,1);
\fill[pattern=Sparse Horizontal Lines] (1,0) rectangle (2,2);
\fill[pattern=Sparse Horizontal Lines] (2,0) rectangle (3,3);
\fill[pattern=Sparse Horizontal Lines] (3,0) rectangle (4,4);
\fill[pattern=Sparse Horizontal Lines] (4,0) rectangle (5,5);
\fill[pattern=Sparse Horizontal Lines] (5,0) rectangle (6,1);
\fill[pattern=Sparse Horizontal Lines] (6,0) rectangle (7,2);
\fill[pattern=Sparse Horizontal Lines] (7,0) rectangle (8,1);

\draw[line width=2pt, blue]  (0,1) -- (1,1) -- (1,2) -- (2,2) -- (2,3) -- (3,3) -- (3,4) -- (4,4) -- (4,5) -- (5,5) -- (5,1) -- (6,1) -- (6,2) -- (7,2) -- (7,1) -- (8,1) ;

\fill[pattern=Sparse North East Lines] (0,0) rectangle (1,1);
\fill[pattern=Sparse North East Lines] (1,0) rectangle (2,2);
\fill[pattern=Sparse North East Lines] (2,0) rectangle (3,3);
\fill[pattern=Sparse North East Lines] (3,0) rectangle (4,4);
\fill[pattern=Sparse North East Lines] (4,0) rectangle (5,5);
\fill[pattern=Sparse North East Lines] (5,0) rectangle (6,2);
\fill[pattern=Sparse North East Lines] (6,0) rectangle (7,3);
\fill[pattern=Sparse North East Lines] (7,0) rectangle (8,4);

\draw[line width=0.8pt, red] (0,1) -- (1,1) -- (1,2) -- (2,2) -- (2,3) -- (3,3) -- (3,4) -- (4,4) -- (4,5) -- (5,5) -- (5,2) -- (6,2) -- (6,3) -- (7,3) -- (7,4) -- (8,4);

\draw[-Stealth, very thick ] (0,0) -- (9,0)  node [xshift=-0.2cm, yshift=0.4cm] {$n$} node [xshift=-9cm, yshift=-0.4cm] {$0$}  node [xshift=-8cm, yshift=-0.4cm] {$1$}  node [xshift=-7cm, yshift=-0.4cm] {$2$}  node [xshift=-6cm, yshift=-0.4cm] {$3$}  node [xshift=-5cm, yshift=-0.4cm] {$4$}  node [xshift=-4cm, yshift=-0.4cm] {$5$}  node [xshift=-3cm, yshift=-0.4cm] {$6$}  node [xshift=-2cm, yshift=-0.4cm] {$7$}  node [xshift=-1cm, yshift=-0.4cm] {$8$}  ;

\draw[-Stealth, very thick ] (0,0) -- (0,6) node [xshift=0.7cm, yshift=-0.2cm] {$AI(n)$}  node [xshift=0.7cm, yshift=-0.7cm] {$A(n)$} node [xshift=0.7cm, yshift=-1.2cm] {$I(n)$} node [xshift=-0.3cm, yshift=-5cm] {$1$} node [xshift=-0.3cm, yshift=-6cm] {$0$};

\draw (1,-0.1) -- (1,0.1);
\draw (2,-0.1) -- (2,0.1);
\draw (3,-0.1) -- (3,0.1);
\draw (4,-0.1) -- (4,0.1);
\draw (5,-0.1) -- (5,0.1);
\draw (6,-0.1) -- (6,0.1);
\draw (7,-0.1) -- (7,0.1);
\draw (8,-0.1) -- (8,0.1);

\draw (-0.1,1) -- (0.1,1);
\draw (-0.1,2) -- (0.1,2);
\draw (-0.1,3) -- (0.1,3);
\draw (-0.1,4) -- (0.1,4);
\draw (-0.1,5) -- (0.1,5);

\node[xshift=2cm,yshift=5.6cm] {AoAI};
\node[xshift=2cm,yshift=5.1cm] {AoA};
\node[xshift=2cm,yshift=4.6cm] {AoI};

\fill[pattern=North East Lines] (3.5-1,6.3-0.5-0.4) rectangle (4.5-1,6.2-0.4);
\fill[pattern=Horizontal Lines] (3.5-1,5.3-0.4) rectangle (4.5-1,6.2-0.5-0.4);
\fill[pattern=Vertical Lines] (3.5-1,5.3-0.5-0.4) rectangle (4.5-1,5.2-0.4);

\draw[line width = 2pt, red] (3.5-1,6.2-0.4) -- (4.5-1,6.2-0.4);
\draw[line width = 2pt, blue]  (3.5-1,6.2-0.5-0.4) rectangle (4.5-1,6.2-0.5-0.4);
\draw[line width = 2pt, green] (3.5-1,5.2-0.4) rectangle (4.5-1,5.2-0.4);

\draw (1.5,4.7-0.4) rectangle (3.7,6.3-0.4);

\end{tikzpicture}}}
\caption{Geo/Geo/1 under LCFS policy with a controller.}
\label{fig:SamplePathInfiniteControllerLCFS}
\vspace{10pt}
\end{subfigure}

\begin{subfigure}[b]{0.9\linewidth}
\centering 
\scalebox{0.8}{ \boldmath{
\begin{tikzpicture}

\fill[pattern=Sparse Vertical Lines] (0,0) rectangle (1,1);
\fill[pattern=Sparse Vertical Lines] (1,0) rectangle (2,2);
\fill[pattern=Sparse Vertical Lines] (2,0) rectangle (3,1);
\fill[pattern=Sparse Vertical Lines] (3,0) rectangle (4,2);
\fill[pattern=Sparse Vertical Lines] (4,0) rectangle (5,1);
\fill[pattern=Sparse Vertical Lines] (5,0) rectangle (6,2);
\fill[pattern=Sparse Vertical Lines] (6,0) rectangle (7,3);
\fill[pattern=Sparse Vertical Lines] (7,0) rectangle (8,4);

\draw[line width=4pt, green]  (0,1) -- (1,1) -- (1,2) -- (2,2) -- (2,1) -- (3,1) -- (3,2) -- (4,2) -- (4,1) -- (5,1) -- (5,2) -- (6,2) -- (6,3) -- (7,3) -- (7,4) -- (8,4);

\fill[pattern=Sparse Horizontal Lines] (0,0) rectangle (1,1);
\fill[pattern=Sparse Horizontal Lines] (1,0) rectangle (2,2);
\fill[pattern=Sparse Horizontal Lines] (2,0) rectangle (3,3);
\fill[pattern=Sparse Horizontal Lines] (3,0) rectangle (4,4);
\fill[pattern=Sparse Horizontal Lines] (4,0) rectangle (5,5);
\fill[pattern=Sparse Horizontal Lines] (5,0) rectangle (6,1);
\fill[pattern=Sparse Horizontal Lines] (6,0) rectangle (7,2);
\fill[pattern=Sparse Horizontal Lines] (7,0) rectangle (8,3);

\draw[line width=2pt, blue]  (0,1) -- (1,1) -- (1,2) -- (2,2) -- (2,3) -- (3,3) -- (3,4) -- (4,4) -- (4,5) -- (5,5) -- (5,1) -- (6,1) -- (6,2) -- (7,2) -- (7,3) -- (8,3) ;

\fill[pattern=Sparse North East Lines] (0,0) rectangle (1,1);
\fill[pattern=Sparse North East Lines] (1,0) rectangle (2,2);
\fill[pattern=Sparse North East Lines] (2,0) rectangle (3,3);
\fill[pattern=Sparse North East Lines] (3,0) rectangle (4,4);
\fill[pattern=Sparse North East Lines] (4,0) rectangle (5,5);
\fill[pattern=Sparse North East Lines] (5,0) rectangle (6,2);
\fill[pattern=Sparse North East Lines] (6,0) rectangle (7,3);
\fill[pattern=Sparse North East Lines] (7,0) rectangle (8,4);

\draw[line width=0.8pt, red] (0,1) -- (1,1) -- (1,2) -- (2,2) -- (2,3) -- (3,3) -- (3,4) -- (4,4) -- (4,5) -- (5,5) -- (5,2) -- (6,2) -- (6,3) -- (7,3) -- (7,4) -- (8,4);

\draw[-Stealth, very thick ] (0,0) -- (9,0)  node [xshift=-0.2cm, yshift=0.4cm] {$n$} node [xshift=-9cm, yshift=-0.4cm] {$0$}  node [xshift=-8cm, yshift=-0.4cm] {$1$}  node [xshift=-7cm, yshift=-0.4cm] {$2$}  node [xshift=-6cm, yshift=-0.4cm] {$3$}  node [xshift=-5cm, yshift=-0.4cm] {$4$}  node [xshift=-4cm, yshift=-0.4cm] {$5$}  node [xshift=-3cm, yshift=-0.4cm] {$6$}  node [xshift=-2cm, yshift=-0.4cm] {$7$}  node [xshift=-1cm, yshift=-0.4cm] {$8$} ;

\draw[-Stealth, very thick ] (0,0) -- (0,6) node [xshift=0.7cm, yshift=-0.2cm] {$AI(n)$}  node [xshift=0.7cm, yshift=-0.7cm] {$A(n)$} node [xshift=0.7cm, yshift=-1.2cm] {$I(n)$} node [xshift=-0.3cm, yshift=-5cm] {$1$} node [xshift=-0.3cm, yshift=-6cm] {$0$};

\draw (1,-0.1) -- (1,0.1);
\draw (2,-0.1) -- (2,0.1);
\draw (3,-0.1) -- (3,0.1);
\draw (4,-0.1) -- (4,0.1);
\draw (5,-0.1) -- (5,0.1);
\draw (6,-0.1) -- (6,0.1);
\draw (7,-0.1) -- (7,0.1);
\draw (8,-0.1) -- (8,0.1);

\draw (-0.1,1) -- (0.1,1);
\draw (-0.1,2) -- (0.1,2);
\draw (-0.1,3) -- (0.1,3);
\draw (-0.1,4) -- (0.1,4);
\draw (-0.1,5) -- (0.1,5);

\node[xshift=6cm,yshift=5.6cm] {AoAI};
\node[xshift=6cm,yshift=5.1cm] {AoA};
\node[xshift=6cm,yshift=4.6cm] {AoI};

\fill[pattern=North East Lines] (3.5+3,6.3-0.5-0.4) rectangle (4.5+3,6.2-0.4);
\fill[pattern=Horizontal Lines] (3.5+3,5.3-0.4) rectangle (4.5+3,6.2-0.5-0.4);
\fill[pattern=Vertical Lines] (3.5+3,5.3-0.5-0.4) rectangle (4.5+3,5.2-0.4);

\draw[line width = 2pt, red] (3.5+3,6.2-0.4) -- (4.5+3,6.2-0.4);
\draw[line width = 2pt, blue]  (3.5+3,6.2-0.5-0.4) rectangle (4.5+3,6.2-0.5-0.4);
\draw[line width = 2pt, green] (3.5+3,5.2-0.4) rectangle (4.5+3,5.2-0.4);

\draw (1.5+4,4.7-0.4) rectangle (3.7+4,6.3-0.4);

\end{tikzpicture}}} 
\caption{Geo/Geo/1/1 with a controller.}
\label{fig:SamplePathBufferController}
\end{subfigure}

\begin{subfigure}[b]{0.9\linewidth}

\centering

\caption{The events and processes.}
\begin{tabular}{|c|c|c|c|c|c|c|c|}
\hline
$n$ & 1 & 2 & 3 & 4 & 5 & 6 & 7  \\ \hline
$\Lambda_1(n)$  & 0 & 1 & 0 & 1 & 0 & 0 & 0  \\ \hline
$\Lambda_2(n)$ & 0 & 0 & 0 & 0 & 1 & 0 & 1  \\ \hline
$Q_{\infty}(n)$ & 0 & 1 & 1 & 2 & 1 & 1 & 0   \\ \hline
$Q_1(n)$ & 0 & 1 & 1 & 1 & 0 & 0 & 0   \\ \hline
\end{tabular}
\label{Table:SamplePathInfiniteControllerFCFS}
\end{subfigure}

\caption{Sample paths and their corresponding events and processes for the concurrent evolution of the AoI, AoA, and AoAI, for the three scenarios mentioned on each figure.}
\end{figure}

\begin{figure}[h!]
\centering
\begin{subfigure}[b]{0.9\linewidth}
\centering 
\scalebox{0.8}{ \boldmath{
\begin{tikzpicture}

\fill[pattern=Sparse Vertical Lines] (0,0) rectangle (1,1);
\fill[pattern=Sparse Vertical Lines] (1,0) rectangle (2,2);
\fill[pattern=Sparse Vertical Lines] (2,0) rectangle (3,1);
\fill[pattern=Sparse Vertical Lines] (3,0) rectangle (4,2);
\fill[pattern=Sparse Vertical Lines] (4,0) rectangle (5,3);
\fill[pattern=Sparse Vertical Lines] (5,0) rectangle (6,4);
\fill[pattern=Sparse Vertical Lines] (6,0) rectangle (7,1);
\fill[pattern=Sparse Vertical Lines] (7,0) rectangle (8,2);

\draw[line width=4pt, green]  (0,1) -- (1,1) -- (1,2) -- (2,2) -- (2,1) -- (3,1) -- (3,2) -- (4,2) -- (4,3) -- (5,3) -- (5,4) -- (6,4) -- (6,1) -- (7,1) -- (7,2) -- (8,2);

\fill[pattern=Sparse Horizontal Lines] (0,0) rectangle (1,1);
\fill[pattern=Sparse Horizontal Lines] (1,0) rectangle (2,2);
\fill[pattern=Sparse Horizontal Lines] (2,0) rectangle (3,3);
\fill[pattern=Sparse Horizontal Lines] (3,0) rectangle (4,4);
\fill[pattern=Sparse Horizontal Lines] (4,0) rectangle (5,1);
\fill[pattern=Sparse Horizontal Lines] (5,0) rectangle (6,2);
\fill[pattern=Sparse Horizontal Lines] (6,0) rectangle (7,1);
\fill[pattern=Sparse Horizontal Lines] (7,0) rectangle (8,2);

\draw[line width=2pt, blue]  (0,1) -- (1,1) -- (1,2) -- (2,2) -- (2,3) -- (3,3) -- (3,4) -- (4,4) -- (4,1) -- (5,1) -- (5,2) -- (6,2) -- (6,1) -- (7,1) -- (7,2) -- (8,2);

\fill[pattern=Sparse North East Lines] (0,0) rectangle (1,1);
\fill[pattern=Sparse North East Lines] (1,0) rectangle (2,2);
\fill[pattern=Sparse North East Lines] (2,0) rectangle (3,3);
\fill[pattern=Sparse North East Lines] (3,0) rectangle (4,4);
\fill[pattern=Sparse North East Lines] (4,0) rectangle (5,3);
\fill[pattern=Sparse North East Lines] (5,0) rectangle (6,4);
\fill[pattern=Sparse North East Lines] (6,0) rectangle (7,1);
\fill[pattern=Sparse North East Lines] (7,0) rectangle (8,2);

\draw[line width=0.8pt, red] (0,1) -- (1,1) -- (1,2) -- (2,2) -- (2,3) -- (3,3) -- (3,4) -- (4,4) -- (4,3) -- (5,3) -- (5,4) -- (6,4) -- (6,1) -- (7,1) -- (7,2) -- (8,2);

\draw[-Stealth, very thick ] (0,0) -- (9,0)  node [xshift=-0.2cm, yshift=0.4cm] {$n$} node [xshift=-9cm, yshift=-0.4cm] {$0$}  node [xshift=-8cm, yshift=-0.4cm] {$1$}  node [xshift=-7cm, yshift=-0.4cm] {$2$}  node [xshift=-6cm, yshift=-0.4cm] {$3$}  node [xshift=-5cm, yshift=-0.4cm] {$4$}  node [xshift=-4cm, yshift=-0.4cm] {$5$}  node [xshift=-3cm, yshift=-0.4cm] {$6$}  node [xshift=-2cm, yshift=-0.4cm] {$7$}  node [xshift=-1cm, yshift=-0.4cm] {$8$};

\draw[-Stealth, very thick ] (0,0) -- (0,6) node [xshift=0.7cm, yshift=-0.2cm] {$AI(n)$}  node [xshift=0.7cm, yshift=-0.7cm] {$A(n)$} node [xshift=0.7cm, yshift=-1.2cm] {$I(n)$} node [xshift=-0.3cm, yshift=-5cm] {$1$} node [xshift=-0.3cm, yshift=-6cm] {$0$};

\draw (1,-0.1) -- (1,0.1);
\draw (2,-0.1) -- (2,0.1);
\draw (3,-0.1) -- (3,0.1);
\draw (4,-0.1) -- (4,0.1);
\draw (5,-0.1) -- (5,0.1);
\draw (6,-0.1) -- (6,0.1);
\draw (7,-0.1) -- (7,0.1);
\draw (8,-0.1) -- (8,0.1);

\draw (-0.1,1) -- (0.1,1);
\draw (-0.1,2) -- (0.1,2);
\draw (-0.1,3) -- (0.1,3);
\draw (-0.1,4) -- (0.1,4);
\draw (-0.1,5) -- (0.1,5);

\node[xshift=3cm,yshift=5.5cm] {AoAI};
\node[xshift=3cm,yshift=5cm] {AoA};
\node[xshift=3cm,yshift=4.5cm] {AoI};

\fill[pattern=North East Lines] (3.5,6.3-1) rectangle (4.5,6.7-1);
\fill[pattern=Horizontal Lines] (3.5,5.8-1) rectangle (4.5,6.2-1);
\fill[pattern=Vertical Lines] (3.5,5.3-1) rectangle (4.5,5.7-1);

\draw[line width = 2pt, red] (3.5,6.7-1) -- (4.5,6.7-1);
\draw[line width = 2pt, blue]  (3.5,6.2-1) rectangle (4.5,6.2-1);
\draw[line width = 2pt, green] (3.5,5.7-1) rectangle (4.5,5.7-1);

\draw (2.5,5.2-1) rectangle (4.7,6.8-1);

\end{tikzpicture}}} 
\caption{Geo/Geo/1/1 with a battery.}
\label{fig:SamplePathBufferBattery}
\end{subfigure}

\begin{subfigure}[b]{0.9\linewidth}
\centering
\caption{The events and processes.}
\label{Table:SamplePathBufferBattery}
\renewcommand{\arraystretch}{1} 
\begin{tabular}{|c|c|c|c|c|c|c|c|}
\cline{1-8}
$n$ & $1$ & $2$ & $3$ & $4$ & $5$ & $6$ & $7$ \\ \hline
$\Lambda_1(n)$  &0 &1 &0 &0 &0 &1 &0 \\ \hline
$\Lambda_2(n)$ &0 &0 &0 &1 &1 &0 &0 \\ \hline
$Q_{B}(n)$ &0 &1 &1 &0 &0 &0 &0 \\ \hline
$B(n)$ &0 &0 &0 &0 &1 &0 &0 \\ \hline
\end{tabular}
\end{subfigure}
\caption{A sample path and its corresponding events and processes for the concurrent evolution of the AoI, AoA, and AoAI for the Geo/Geo/1/1 with a battery.}
\end{figure}

\section{The relations between the AoI, AoA, and AoAI} 

In this section, we analyze the evolution of AoI, AoA, and AoAI and their averages in comparison to each other for the scenarios and cases of the system model. In this regard, they fall into two categories based on the queue sizes of infinity and $1$.

\subsection{The Scenarios with the Geo/Geo/1} \label{RelationsInQueueOnes}
Figs. \ref{fig:SamplePathInfiniteControllerFCFS} and \ref{fig:SamplePathInfiniteControllerLCFS} illustrate sample paths demonstrating the concurrent evolution of AoI, AoA, and AoAI for the scenarios of Geo/Geo/1 under two policies of FCFS and LCFS, respectively. The events of $\Lambda_1(n)$ and $\Lambda_2(n)$, and the process $Q_{\infty}(n)$ are also presented for the illustrated time slots in Table \ref{Table:SamplePathInfiniteControllerFCFS}.
For these scenarios, we always have $AI(n) \geq I(n)$ and $AI(n) \geq A(n)$: AoAI is an upper bound for both the AoI and AoA. The former is because the reception of a data packet is always prior to its utilization for an actuation. The latter is because the age of a packet cannot be lower than $1$ at the end of the actuation time slot. As a result then $\overline{AI} \geq  \bar{I}$ and $\overline{AI} \geq \bar{A}$. Note that $A(n) \lesseqqgtr I(n)$, as its three modes can be noticed and tracked in Figs. \ref{fig:SamplePathInfiniteControllerFCFS} and \ref{fig:SamplePathInfiniteControllerLCFS}. Nevertheless, ultimately, $\bar{I} = \bar{A}$. The intuition is discussed in Remark \ref{TheIntuituionForIAequality}. In conclusion, we have $\bar{I} = \bar{A} \leq \overline{AI}$.
Also, in extreme cases:
\begin{itemize}
    \item $\lambda_1 \rightarrow \lambda_2 \rightarrow 1$ : \  $\overline{AI} \rightarrow \bar{A} =  \bar{I} \rightarrow 1 $,
    \item $\lambda_2 \rightarrow 1$ : \  $\overline{AI} \rightarrow \bar{A}= \bar{I} = \displaystyle{\frac{1}{\lambda_1}} \geq 1 $.
\end{itemize}

\subsection{The Scenarios with the Geo/Geo/1/1}
Figs. \ref{fig:SamplePathBufferController} and \ref{fig:SamplePathBufferBattery} are sample paths of the concurrent evolution of AoI, AoA, and AoAI, for the scenarios Geo/Geo/1/1, with a controller and with a battery, respectively. Also, the events $\Lambda_1(n)$ and $\Lambda_2(n)$, and the process $Q_{1}(n)$ are presented in Table \ref{Table:SamplePathInfiniteControllerFCFS} for the Geo/Geo/1/1 with controller, and the events $\Lambda_1(n)$ and $\Lambda_2(n)$ and the processes $Q_{B}(n)$ and $B(n)$ for the scenario Geo/Geo/1/1 with a battery are presented in Table \ref{Table:SamplePathBufferBattery}. Most of the relations are the same as section \ref{RelationsInQueueOnes}: $AI(n) \geq I(n)$ and thus $AI(n) \geq A(n)$ and $\overline{AI} \geq  \bar{I}$ and $\overline{AI} \geq \bar{A}$, and also $A(n) \lesseqqgtr I(n)$, as again its three modes can be noticed and tracked in Figs. \ref{fig:SamplePathBufferController} and \ref{fig:SamplePathBufferBattery}. The only difference is $\bar{I} \leq \bar{A}$. This is because each data packet resets the AoI to $1$, but not each of those data packets resets the AoA to $1$ as well. This is because there is a single-size queue to store data packets, and thus, there are not-yet-actuated packets that get disposed of upon the arrival of new packets. In conclusion, $\bar{I} \leq \bar{A} \leq \overline{AI}$.
Also, in extreme cases:
\begin{itemize}
    \item $\lambda_1 \rightarrow 1$ : \  $\overline{AI} \rightarrow \bar{A} \rightarrow \displaystyle{\frac{1}{\lambda_2}} \geq \bar{I} \rightarrow 1 $,
    \item  $\lambda_2 \rightarrow 1$ : \  $\overline{AI} \rightarrow \bar{A} \rightarrow \bar{I} \rightarrow \displaystyle{\frac{1}{\lambda_1}} \geq 1 $,
    \item $\lambda_1 \rightarrow 1, \ \lambda_2 \rightarrow 1$ : \ $\overline{AI} \rightarrow \bar{A} \rightarrow \bar{I} \rightarrow  1 $.
\end{itemize}

\section{Numerical Results} \label{Numerical and Simulation Results}

The results in Table \ref{tableofFormulas}, except for two of them, are monotonically decreasing with respect to $\lambda_1$ and $\lambda_2$ (increasing $\lambda_1$ and $\lambda_2$ results in a decrease in the metric), which is intuitive. The exceptions are: first, the average AoAI for the Geo/Geo/1 under the FCFS policy with a controller, i.e., (\ref{AoAI_average_InfCache}), only against $\lambda_1$, and second, the average AoA for the Geo/Geo/1/1 with a battery, i.e., (\ref{Average_AoA_Battery}), against both $\lambda_1$ and $\lambda_2$. These two display critical points in the domain of $0 \leq \lambda_1, \ \lambda_2 \leq 1$.

These critical points can be obtained by taking derivatives from (\ref{AoAI_average_InfCache}) and (\ref{Average_AoA_Battery}) with respect to $\lambda_1$($\lambda_2$) and obtaining the valid roots that are functions of $\lambda_2$($\lambda_1$)\footnote{We refrain from bringing the expressions due to their complexity and size.}.

\begin{figure}[htbp]
     \centering
     \begin{subfigure}[b]{0.49\textwidth}
         \centering
         \includegraphics[width=\linewidth]{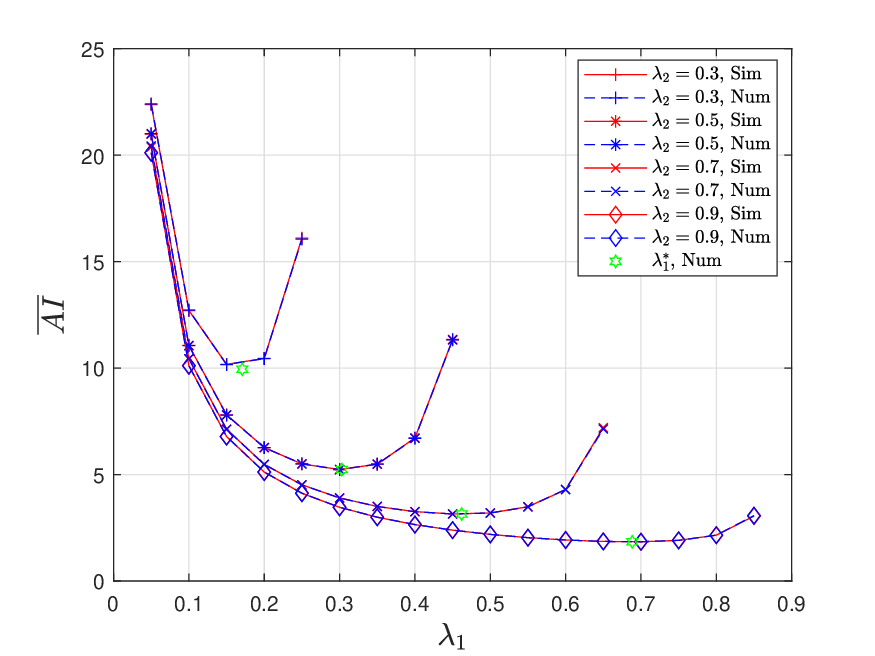}
        \caption{The average AoAI, for the Geo/Geo/1 under FCFS policy with a controller, for four different values of $\lambda_2$ against $\lambda_1$ and the optimum values of $\lambda_1$.}
         \label{AoAI_Infinite_q2_03579}
     \end{subfigure}
     \hfill 
     \begin{subfigure}[b]{0.49\textwidth}
         \centering
         \includegraphics[width=\linewidth]{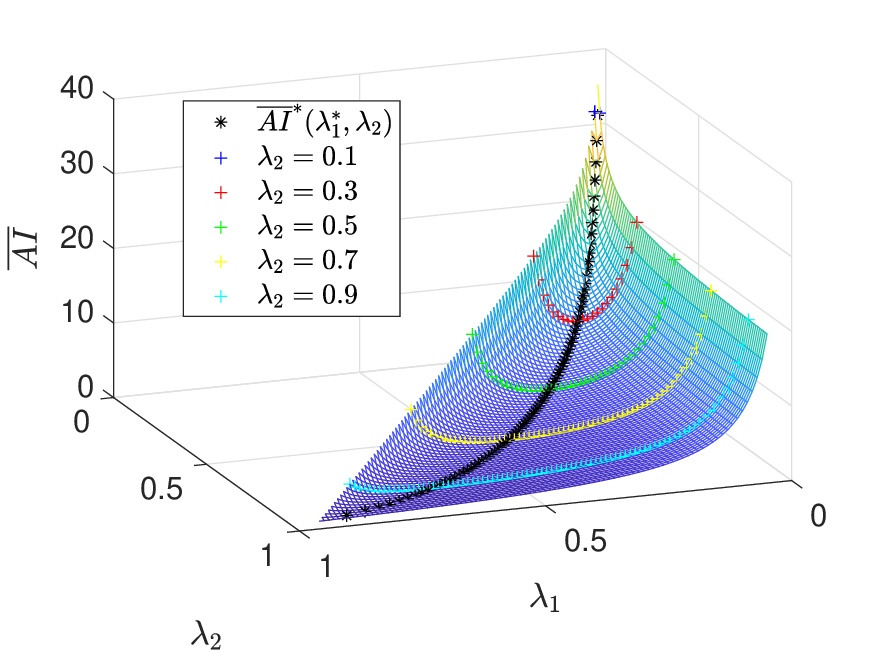}
        \caption{The Average AoAI, for the Geo/Geo/1 under FCFS policy with the controller, against $\lambda_1$ and $\lambda_2$ and optimum values of the $\bar{AI}$ and $\lambda_1$ against all values of $\lambda_2$.}
         \label{AoAI_3D_OPT}
     \end{subfigure}
     
     \caption{The average AoAI, for the Geo/Geo/1 under FCFS policy with a controller.}
     \label{fig:combined_aoai_plots}
\end{figure}

\subsection{The Average AoAI for the Geo/Geo/1 under the FCFS Policy with a Controller}
As mentioned, (\ref{AoAI_average_InfCache}) has critical points in the range $0 \leq \lambda_1 \leq 1$ when $\lambda_2$ is fixed, but not the other way around.
Fig. \ref{AoAI_Infinite_q2_03579} illustrates the results of the average AoAI for the Geo/Geo/1 under the FCFS policy with a controller for constant $\lambda_2$s and varying $\lambda_1$s in addition to each $\lambda_1^*$. Fig. \ref{AoAI_3D_OPT} shows the results against $\lambda_1$ and $\lambda_2$. Upon the surface, the curves of five constant $\lambda_2$ values are highlighted. Also, a black curve showing the minimum points against $\lambda_2$ values is drawn.

As can be seen, increasing the frequency of transmission counterintuitively does not always cause the average AoAI to decrease.

\begin{figure}[htbp]
    \centering
    \begin{subfigure}[b]{0.49\linewidth}
        \centering
        \includegraphics[width=\linewidth]{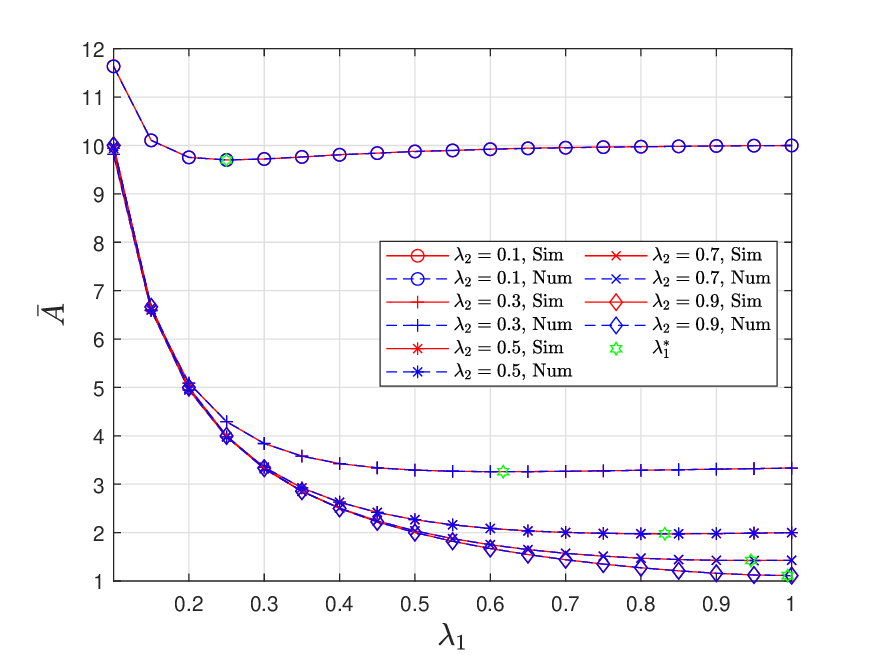}
        \caption{The average AoA.}
        \label{AoABattery_q2_0.13579}
    \end{subfigure}
    \hfill 
    \begin{subfigure}[b]{0.49\linewidth}
        \centering
        \includegraphics[width=\linewidth]{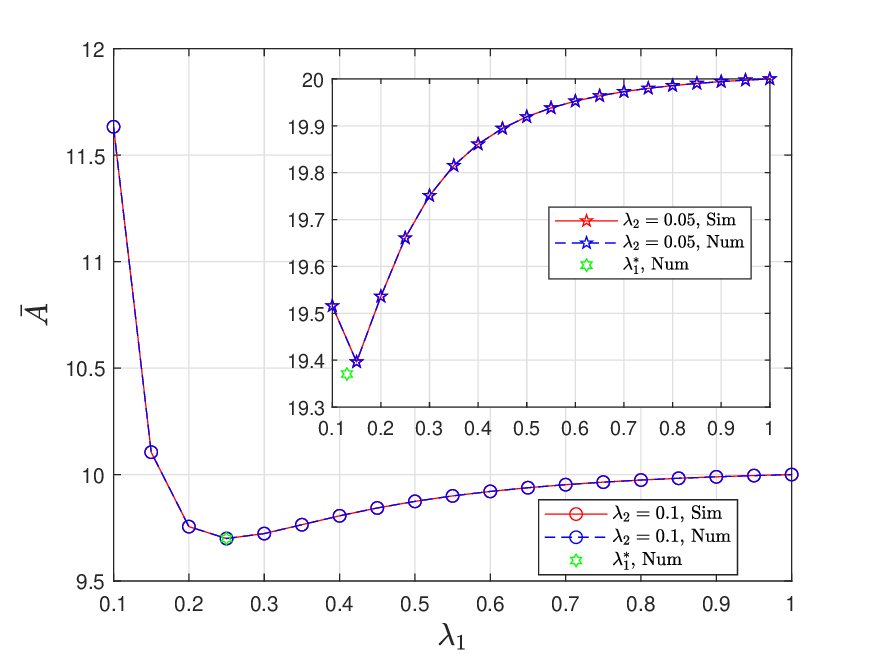}
        \caption{The average AoA.}
        \label{AoABatteryq2_0.1_0.05}
    \end{subfigure}
    
    \caption{The average AoA, for the Geo/Geo/1/1 with a battery, for different values of $\lambda_2$ against $\lambda_1$ and the optimum values of $\lambda_1$.}
    \label{fig:combined_battery_aoa}
\end{figure}

\begin{figure}[htbp]
    \centering
    \includegraphics[width=0.5\linewidth]{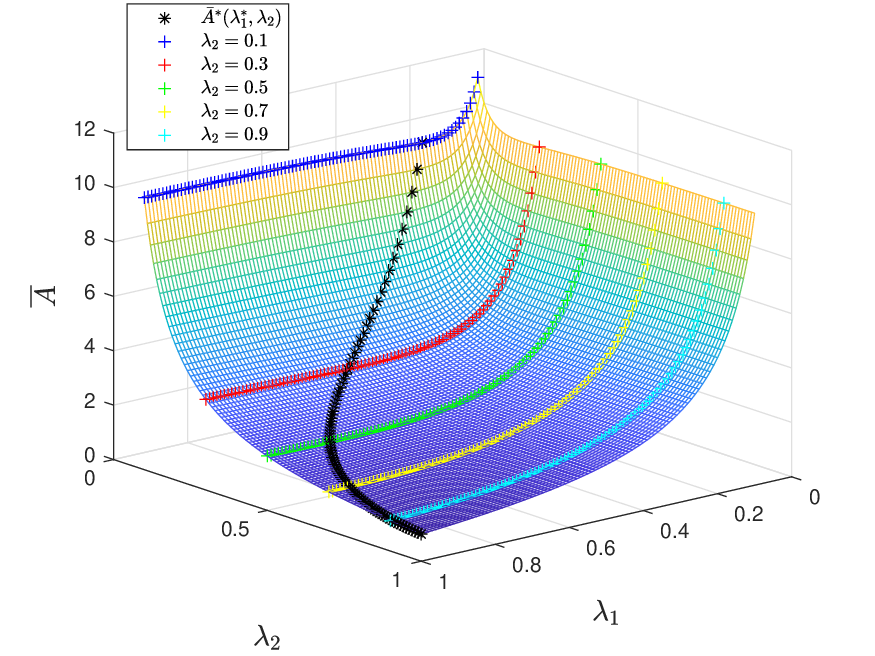}
     \caption{The Average AoA, for the Geo/Geo/1/1 with a battery, against $\lambda_1$ and $\lambda_2$ and optimum values of the $\bar{A}$ and $\lambda_1$ against all values of $\lambda_2$, with five different values of $\lambda_2$ highlighted.}
     \label{AoA_3D_OPT}
\end{figure}

This behavior can be explained as follows: \textit{Overwhelming the queue under FCFS policy, in which all the packets are due to be actuated in order, puts fresh packets backlogged behind the stale not-yet-actuated ones and results in higher average AoAI values. Therefore, in this scenario, an optimum rate of updating is more effective in improving the timeliness of actions when the age of the actuated packets is a critical factor.}

\subsection{The Average AoA for the Geo/Geo/1/1 with a Battery}
We observe that (\ref{Average_AoA_Battery}) exhibits approximately symmetric behavior with respect to its variables $\lambda_1$ and $\lambda_2$. It means that, considering $\bar{A}(\lambda_1,\lambda_2)$, we find that $\bar{A}(\hat{\lambda}_1,\hat{\lambda}_2) \approx \bar{A}(\hat{\lambda}_2,\hat{\lambda}_1)$. Given the symmetry, we will present $\bar{A}$ for various values of $\lambda_2$ against $\lambda_1$.

In Fig. \ref{AoABattery_q2_0.13579}, the average AoA for different values of $\lambda_2$ against $\lambda_1$ and the optimum points, $\lambda_1^*$, are depicted. Similarly, Fig. \ref{AoABatteryq2_0.1_0.05} displays average AoA for two low values of $\lambda_2$, more clearly.
Fig. \ref{AoA_3D_OPT} shows the results against $\lambda_1$ and $\lambda_2$. Upon the surface, we also highlight the curves of five constant $\lambda_2$ values. Also, we illustrate the black curve showing the minimum points against $\lambda_2$ values.

We observe that increasing $\lambda_1$, specifically for lower values of $\lambda_2$, can lead to a rise in the average AoA, which is counterintuitive.

This behavior is also explained as follows: \textit{Updating (charging) the single-size queue (battery) too frequently can consume the battery (data) resource too often and leaves the next data packets (energy units) not-yet-actuated (unutilized) for a while, increasing the average AoA values, especially when energy (data) availability is limited. Therefore, less frequent updates (charging) can be more effective in improving the timeliness of actions when the age of the actuated packets is not a critical factor.}

\section{Conclusions}
This paper examined the timeliness of actions in communication systems where actuation is constrained by controller permissions or energy availability. We showed that information-centric metrics such as the Age of Information (AoI) are insufficient when actions are the ultimate objective. To address this limitation, we revisited the Age of Actuation (AoA). We introduced the \emph{Age of Actuated Information (AoAI)}, which explicitly accounts for the age of the data at the moment it is used for actuation.

Closed-form expressions for the average AoA and AoAI were derived for several discrete-time system models with different queueing and actuation configurations.
We showed that AoA and AoAI coincide in instantaneous actuation systems but diverge when data storage is possible, revealing that actions with similar execution timing may differ significantly in effectiveness due to differences in data freshness. Moreover, our results uncovered counterintuitive regimes in which increasing update or actuation rates worsen action timeliness for both metrics.

In addition, the analysis yielded a novel characterization of the steady-state distribution of a Geo/Geo/1 queue under the FCFS discipline, expressed in terms of the queue length and the age of the head-of-line packet. Overall, this work contributes action-aware timeliness metrics and analytical tools that support the design of semantic and goal-oriented communication systems, and enables future extensions to more general network and control settings.

\begin{appendices}

\section{Proof of Theorem \ref{theo:AoA}} \label{App:AoA}

\subsection{Proof of Theorem \ref{theo:AoA}\ref{theo:The Average AoA Infinite}} \label{Proof of The Average AoA Infinite}

We cast a two-dimensional Markov chain as $(A,Q)$ in which $A=\{1,2,3,\hdots\}$ is the AoA at the end of each time slot and $Q=\{0,1,2,\hdots\}$ is the number of the packets left in the queue at the end of each time slot. Thus, the states would be $\{(1,0),(1,1),(1,2),\hdots,(2,0),(2,1),(2,2),\hdots\}$. The transition probabilities will be given in tabular form
\[
\begin{array}{c|ccccc}
    & (1,0) & (1,1) & (1,2) & (1,3) & \cdots \\
\hline
(A,0) &w & 0 & 0 & 0  & \cdots \\
(A,1) &y &w & 0 & 0 &  \cdots \\
(A,2) & 0 & y &w & 0 & \cdots \\
(A,3) & 0&  0&y &w & \cdots \\
\vdots & \vdots & \vdots & \vdots & \vdots & \ddots \\
\end{array}
,\
\begin{array}{c|cccccc}
    & (A+1,0) & (A+1,1) & (A+1,2) & (A+1,3)  & \cdots \\
\hline
(A,0) & y+z  & x & 0  & 0 & \cdots \\
(A,1) & 0 &z &x & 0  &  \cdots \\
(A,2) & 0 & 0 &z &x & \cdots \\
\vdots & \vdots & \vdots & \vdots & \vdots & \ddots 
\end{array}
\]

for $A=\{1,2,3,\cdots\}$, where $w=\lambda_1 \lambda_2, \ x=\lambda_1 \bar{\lambda}_2, \ y=\bar{\lambda}_1 \lambda_2, \ z=\bar{\lambda}_1 \bar{\lambda}_2 $. Then, the whole transition probability matrix can be obtained.
Defining $Q_j=\text{Pr}(Q=j)=\sum_{i=1}^{\infty} V_{i,j}$, the first set of equations for $V_{1,i}$s are 
\begin{equation} \label{For1}
V_{1,j}=wQ_{j}+yQ_{j+1}, \ \forall j \in \{0,1,2,3,\cdots\}.
\end{equation}
Other equations lead to two patterns of equations
\begin{equation}\label{ForOth_1}
V_{i+1,0}=\bar{\lambda}_1 V_{i,0}, \ \forall i \in \{1,2,3,\cdots\},
\end{equation}
\begin{equation} \label{ForOth_2}
V_{i+1,j}=x V_{i,j-1} + z V_{i,j}, \ \forall i,j \in \{1,2,3,\cdots\}.
\end{equation}
By (\ref{ForOth_1}), we can see that
\begin{equation} \label{Q_0_and_lambda_1}
Q_0=\sum_{i=1}^{\infty} V_{i,0}=\frac{V_{1,0}}{\lambda_1}
\end{equation}
Summing (\ref{ForOth_2}) over $i$ for a specific $j$ and also adding (\ref{For1}) for the same $j$ to it, results in
\begin{equation} \label{Cs}
(x+y)Q_{j}=xQ_{j-1}+yQ_{j+1}, \ \forall j \in \{1,2,3,\cdots\}.
\end{equation}
Summing (\ref{Cs}) over $j$ results in $Q_1=\frac{x}{y}Q_0$, by which and using equation of (\ref{Cs}) recursively we can obtain
\begin{equation} \label{C_i_C_0}
Q_{i}=\left(\frac{x}{y}\right)^i Q_{0}, \ \forall \in \{1,2,3,\cdots \}.
\end{equation}
By summing (\ref{C_i_C_0}) recursively for $i$s and equaling to $1$, we can see that
\begin{equation} \label{C_0}
Q_0=1-\frac{x}{y}.
\end{equation}
\begin{remark}
We could have also known (\ref{C_i_C_0}) and (\ref{C_0}) from Geo/Geo/1 \cite[Section 3.4.2, Equation 3.19]{srikant2013communication}, since the queue has $\lambda_1$ as the arrival probability and $\lambda_2$ as the departure probability.
\end{remark}
By (\ref{For1}) and (\ref{Cs}) we can see that
\begin{equation} \label{V_1j_V_10}
V_{1,j}=\left(\frac{x}{y}\right)V_{1,0}
\end{equation}
We define $A_i=\text{Pr}(A=i)=\sum_{j=0}^{\infty} V_{i,j}$. Then, from (\ref{For1}) and (\ref{C_0}) we can obtain
\begin{equation} \label{A_1}
A_1=\sum_{j=1}^{\infty} V_{1,j} =w+y(1-Q_0)=w+x=\lambda_1,
\end{equation}

Summing (\ref{ForOth_1}) over $j$ and adding (\ref{ForOth_2}) for a specific $i$ we can see
\begin{equation} \label{A_i_1_A_i}
A_{i+1}= T \bar{\lambda}_1^{i-1} +\bar{\lambda}_2 A_i ,
\end{equation}
where $T=( \lambda_2 - \lambda_1 + x ) \lambda_1 (1-\frac{x}{y})$. From (\ref{A_1}) and (\ref{A_i_1_A_i}) we can write
\begin{align*}
A_1 &= \lambda_1, \ A_2 = T+ \lambda_1 \bar{\lambda}_2, \ A_3 = T(\bar{\lambda}_1 + \bar{\lambda}_2)+ \lambda_1 \bar{\lambda}_2^2, A_4 = T(\bar{\lambda}_1^2 + \bar{\lambda}_1 \bar{\lambda}_2 + \bar{\lambda}_2^2)+ \lambda_1 \bar{\lambda}_2^3,\\
A_5 &= T(\bar{\lambda}_1^3 + \bar{\lambda}_1^2 \bar{\lambda}_2 + 
\bar{\lambda}_1 \bar{\lambda}_2^2 + \bar{\lambda}_2^3)+ \lambda_1 \bar{\lambda}_2^4, \cdots
\end{align*}
Weighting them by $i$, i.e. multiplying $A_i$ by $i$ and summing all terms, replacing $T$, $x$, and $y$, results in the average AoA
\begin{equation} \label{Average_AoA_InfiniteCahche}
\bar{A}=\left(\frac{1}{\lambda_1 \lambda_2^2} + \frac{1}{\lambda_1^2 \lambda_2}\right) T + \frac{\lambda_1}{\lambda_2^2}= \frac{1}{\lambda_1}. \hfill \qed
\end{equation}

\subsection{Proof of Theorem \ref{theo:AoA}\ref{theo:The Average AoA cache}} \label{Proof of The Average AoA cache}
Consider the two-dimensional Markov chain $(A,Q)$, in which $A \in \{1,2,3,\hdots\}$ is the AoA at the end of the time slot. The process of the queue is represented by $Q(n) \in {0,1}$, where $0$ signifies an empty queue and $1$ a full queue at the end of time slot $n$.

Thus, the states are $\{1,0\}$, $\{1,0\}$, $\{2,0\}$, $\{2,1\}$, $\{3,0\}$, $\{3,1\}$, $\cdots$ . We cannot have $\{1,1\}$, since $A=1$ means an instantaneous data packet has been utilized, and the queue should be empty.
Therefore, the transition probability matrix for this process can be represented As
\begin{equation} \label{TPM_AoA}
\mathbf{P}^{A}=
\begin{bmatrix}
\lambda_1 \lambda_2 & \bar{\lambda}_1  &  \lambda_1 \bar{\lambda}_2  & 0& 0&0& 0&\hdots\\
\lambda_1 \lambda_2 & 0 &0 &\bar{\lambda}_1 &  \lambda_1 \bar{\lambda}_2 &  0& 0& \hdots \\
 \lambda_2 & 0 &0 & 0 &\bar{\lambda}_2 &0 & 0& \hdots \\
\lambda_1 \lambda_2 & 0 &0 &0&0 &\bar{\lambda}_1 &  \lambda_1 \bar{\lambda}_2   & \hdots \\
 \lambda_2 & 0 &0 & 0 &0 & 0 &\bar{\lambda}_2 &  \hdots \\
\vdots & \vdots & \vdots&\vdots&\vdots&\vdots&\vdots& \ddots\\
\end{bmatrix} .
\end{equation}

The probability of occurrence of the state $(A,Q)$ is denoted by $V_{A,Q}$. The steady states vector is $\mathbf{V} = [V_{1,0}, V_{2,0}, V_{2,1}, V_{3,0}, V_{3,1},  \hdots]$. Hence, by the equations $\mathbf{V}\mathbf{P}=\mathbf{V}$ and $\mathbf{V} \mathbf{1}=1$, we can obtain the steady states. Note that $Q_{0}=\text{Pr}(Q=0)$.
The first equation of $\mathbf{V}\mathbf{P}=\mathbf{V}$ yield
\begin{equation} \label{V_101_C_0_1}
V_{1,0}=\lambda_1 \lambda_2 Q_{0} + \lambda_2 Q_{1}
\end{equation}
Writing the other equations of $\mathbf{V}\mathbf{P}=\mathbf{V}$ can provide us with two recursive structural equations
\begin{equation}
V_{i+1,0}=\bar{\lambda}_1  V_{i,0}, \  \forall i \in \{1,2,\hdots\}, \ V_{i+1,1}=\lambda_1 \bar{\lambda}_2  V_{i,0} + \bar{\lambda}_2 V_{i,1}, \  \forall i \in \{1,2,\hdots\}.
\end{equation}
Taking $w=\lambda_1 \lambda_2$, $x=\lambda_1 \bar{\lambda}_2$, $y=\bar{\lambda}_1 \lambda_2$, and $z=\bar{\lambda}_1 \bar{\lambda}_2$, by the recursive equations, we can write all of the state probabilities based on only $V_{1,0}$:
\begin{align*}
V_{1,0}&=V_{1,0}, \ V_{2,0}=(\bar{\lambda}_1) V_{1,0}, \ V_{2,1}=(x) V_{1,0}, V_{3,0}=(\bar{\lambda}_1^2) V_{1,0}, \ V_{3,1}=(x(\bar{\lambda}_1 + \bar{\lambda}_2)) V_{1,0},\\
V_{4,0}&=(\bar{\lambda}_1^3) V_{1,0}, \ V_{4,1}=(x(\bar{\lambda}_1^2 +\bar{\lambda}_1 \bar{\lambda}_2+ \bar{\lambda}_2^2)) V_{1,0}, \cdots  \numberthis \label{StateEquations_AoA_Geo/Geo/1/1_Controller}
\end{align*}
If we sum the corresponding terms, we obtain:
\begin{equation}
Q_{0}=\frac{1}{\lambda_1} V_{1,0}, \quad \quad
Q_{1}=\frac{\bar{\lambda}_2}{\lambda_2} V_{1,0}.
\end{equation}
We know that $Q_0+Q_1=1$. Then
\begin{equation} \label{V_101_V_100_1}
V_{1,0}=\frac{\lambda_1 \lambda_2}{ \lambda_2+ \lambda_1 \bar{\lambda}_2 }.
\end{equation}

If we weight the terms in (\ref{StateEquations_AoA_Geo/Geo/1/1_Controller}) corresponding to their AoA, meaning that we multiply terms $V_{A,Q}$ by $A$, sum all of the weighted terms, we get
\begin{equation} \label{v_10}
\bar{A}=\left( \frac{1}{\lambda_1^2} + x \left(\frac{\bar{\lambda}_1}{\lambda_1^2 \lambda_2} + \frac{1}{\lambda_1 \bar{\lambda}_2 \lambda_2^2} - \frac{1}{\lambda_1 \bar{\lambda}_2} \right)   \right) V_{1,0}.
\end{equation}
Replacing $V_{1,0}$ from (\ref{v_10}) and $x$, we can obtain the expected value of the AoA, i.e., the average AoA:
\begin{equation}
\bar{A}=\frac{ (\lambda_2-1) (\lambda_1^2+ \lambda_1 \lambda_2)   - \lambda_2^2}{\lambda_1 \lambda_2 \left(\lambda_1  ( \lambda_2-1) - \lambda_2\right) }. \hfill \qed
\end{equation}

\subsection{Proof of Theorem \ref{theo:AoA}\ref{theo:The Average AoA battery}} \label{Proof of The Average AoA battery}
Consider the three-dimensional Markov chain $(A,Q,B)$, in which $A \in \{1,2,3,\hdots\}$ is the AoA at the end of the time slot. The state of the queue is represented by $Q(n) \in \{0,1\}$, where $0$ signifies an empty queue and $1$ a full queue at the end of time slot $n$. Similarly, $B(n) \in \{0,1\}$ indicates the battery's process, with $0$ representing an empty battery and $1$ a full battery at the end of time slot $n$.

Thus, the states are $\{1,0,0\}$, $\{1,0,1\}$, $\{2,0,0\}$, $\{2,0,1\}$, $\{2,1,0\}$, $\{3,0,0\}$, $\{3,0,1\}$, $\{3,1,0\}$, $\cdots$ . Obviously, we cannot have states such as $\{A,1,1\}$. Because, as we mentioned earlier, if there are both an energy unit and a data packet, they should already have been utilized for an actuation together. Also, we cannot have $\{1,1,0\}$, since $A=1$ means an instantaneous data packet has been utilized, and the queue should be empty.
Therefore, the transition probability matrix for this process can be represented in (\ref{TPM_Battery_AoA}).

\begin{figure*}
\vspace{1em}
\setcounter{MaxMatrixCols}{30}
\begin{equation} \label{TPM_Battery_AoA}
\resizebox{0.7\textwidth}{!}{$
\mathbf{P}^{A,Ba}=
\begin{bmatrix}
\lambda_1 \lambda_2 & 0  & \bar{\lambda}_1 \bar{\lambda}_2 & \bar{\lambda}_1 \lambda_2 & \lambda_1 \bar{\lambda}_2&0& 0& 0&0& 0&0& 0&\hdots\\
\lambda_1 \bar{\lambda}_2 & \lambda_1 \lambda_2 &0 &\bar{\lambda}_1 &  0 & 0&0&0&0& 0& 0& 0& \hdots \\
\lambda_1 \lambda_2 & 0 &0 & 0 &0 & \bar{\lambda}_1  \bar{\lambda}_2 & \bar{\lambda}_1 \lambda_2 & \lambda_1  \bar{\lambda}_2 &0& 0& 0& 0& \hdots \\
\lambda_1 \bar{\lambda}_2 & \lambda_1 \lambda_2 & 0&0& 0& 0 & \bar{\lambda}_1  &0& 0&0& 0& 0  & \hdots \\
\lambda_2 & 0&0&0&0&0&0& \bar{\lambda}_2& 0&0& 0& 0  & \hdots\\
\lambda_1 \lambda_2 &0&0&0&0&0&0&0&\bar{\lambda}_1  \bar{\lambda}_2 & \bar{\lambda}_1 \lambda_2 & \lambda_1  \bar{\lambda}_2& 0 & \hdots\\
\lambda_1 \bar{\lambda}_2 & \lambda_1 \lambda_2 & 0& 0& 0&0&0& 0& 0 & \bar{\lambda}_1  &0& 0  & \hdots\\
\lambda_2 & 0&0&0&0&0&0&0&0&0& \bar{\lambda}_2& 0 & \hdots\\
\vdots & \vdots & \vdots&\vdots&\vdots&\vdots&\vdots&\vdots&\vdots&\vdots&\vdots&\vdots & \ddots\\
\end{bmatrix}
$}
\end{equation}
\setcounter{MaxMatrixCols}{30}
\begin{equation} \label{TPM_Inf_AoAI}
\resizebox{0.7\textwidth}{!}{$
\mathbf{P}^{AI,\infty}=
\begin{bmatrix}
 \lambda_1 \lambda_2 &  \bar{\lambda}_1 & \lambda_1 \bar{\lambda}_2 & 0& 0&0& 0&0&0&0&0& 0&0&0&0&\hdots\\
\lambda_1 \lambda_2 &      0 &0&\bar{\lambda}_1& \lambda_1  \bar{\lambda}_2 &0 &  0 & 0&0&0&0&0&0&0&0& \hdots \\
 0 &  \bar{\lambda}_1 \lambda_2 &\lambda_1 \lambda_2&0&0&\bar{\lambda}_1 \bar{\lambda}_2  &\lambda_1 \bar{\lambda}_2 & 0& 0   &0& 0&  0&0& 0&  0& \hdots \\
 \lambda_1 \lambda_2 &      0 &0&0 &0&0 &0&\bar{\lambda}_1& \lambda_1  \bar{\lambda}_2 &0 &  0 & 0&0&0&0& \hdots \\
 0 &  \bar{\lambda}_1 \lambda_2 &\lambda_1 \lambda_2&0&0&0&0&0&0&\bar{\lambda}_1 \bar{\lambda}_2  &\lambda_1 \bar{\lambda}_2 & 0& 0   &0& 0&  \hdots \\
 0 & 0&0&  \bar{\lambda}_1 \lambda_2 &\lambda_1 \lambda_2&0&0&0&0&0&0&\bar{\lambda}_1 \bar{\lambda}_2  &\lambda_1 \bar{\lambda}_2 & 0& 0   & \hdots \\
0&0&0&0&0&  \bar{\lambda}_1 \lambda_2 &\lambda_1 \lambda_2&0&0&0&0&0&0&\bar{\lambda}_1 \bar{\lambda}_2  &\lambda_1 \bar{\lambda}_2 & \hdots \\
\lambda_1 \lambda_2 &      0 &0&0 &0&0 &0&0&0 &0 &  0 & 0&0&0&0& \hdots \\
 0 &  \bar{\lambda}_1 \lambda_2 &\lambda_1 \lambda_2&0&0&0&0&0&0&0  &0 & 0& 0   &0& 0&  \hdots \\
0 & 0&0&  \bar{\lambda}_1 \lambda_2 &\lambda_1 \lambda_2&0&0&0&0&0&0&0  &0 & 0& 0   & \hdots \\
0&0&0&0&0&  \bar{\lambda}_1 \lambda_2 &\lambda_1 \lambda_2&0&0&0&0&0&0&0  &0 & \hdots \\
0 &0&0&0&0&0&0&  \bar{\lambda}_1 \lambda_2 &\lambda_1 \lambda_2&0  &0 & 0& 0   &0& 0&  \hdots \\
0 & 0&0&0&0&0&0&0&0&  \bar{\lambda}_1 \lambda_2 &\lambda_1 \lambda_2&0  &0 & 0& 0   & \hdots \\
0&0&0&0&0&0&0&0&0&0&0&  \bar{\lambda}_1 \lambda_2 &\lambda_1 \lambda_2&0&0& \hdots \\
0&0&0&0&0&0&0&0&0&0&0&0&0&  \bar{\lambda}_1 \lambda_2 &\lambda_1 \lambda_2& \hdots \\
\vdots&\vdots&\vdots&\vdots&\vdots&\vdots&\vdots&\vdots&\vdots&\vdots&\vdots&\vdots&\vdots&\vdots&\vdots & \ddots\\ 
\end{bmatrix}
$}
\end{equation}
\setcounter{MaxMatrixCols}{30}
\begin{equation} \label{TPM_1_AoAI}
\resizebox{0.6\textwidth}{!}{$
\mathbf{P}^{AI,1}=
\begin{bmatrix}
\lambda_1 \lambda_2 & \lambda_1 \bar{\lambda}_2 & \bar{\lambda}_1  & 0& 0&0& 0&0&0&0&\hdots\\
\lambda_1 \lambda_2 & 0 &\bar{\lambda}_1 \lambda_2& \lambda_1  \bar{\lambda}_2 &\bar{\lambda}_1 \bar{\lambda}_2 &  0 & 0&0&0&0& \hdots \\
\lambda_1 \lambda_2 &0 &0&  \lambda_1 \bar{\lambda}_2 &0 & \bar{\lambda}_1& 0   &0& 0&  0& \hdots \\
\lambda_1 \lambda_2 &0& \bar{\lambda}_1 \lambda_2& 0 & 0& 0 &  \lambda_1 \bar{\lambda}_2   &\bar{\lambda}_1 \bar{\lambda}_2& 0& 0& \hdots \\
\lambda_1 \lambda_2 &	0&	0&	0&	0&	\bar{\lambda}_1  \lambda_2&  \lambda_1 \bar{\lambda}_2&	0&	 \bar{\lambda}_1 \bar{\lambda}_2&0& \hdots\\
\lambda_1 \lambda_2 &0 &0&0&0& 0&  \lambda_1 \bar{\lambda}_2  &0& 0 &\bar{\lambda}_1     & \hdots \\
\vdots&\vdots&\vdots&\vdots&\vdots&\vdots&\vdots&\vdots&\vdots&\vdots & \ddots\\ 
\end{bmatrix}
$}
\end{equation}
\end{figure*}
\begin{figure*}
\vspace{-1.5em}
\setcounter{MaxMatrixCols}{30}
\begin{equation} \label{TPM_Battery_AoAI}
\resizebox{0.7\textwidth}{!}{$
\mathbf{P}^{AI,Ba}=
\begin{bmatrix}
\lambda_1 \lambda_2 & 0  & \lambda_1 \bar{\lambda}_2 & \bar{\lambda}_1 \bar{\lambda}_2 & \bar{\lambda}_1 \lambda_2&0& 0& 0&0& 0&0& 0&0&0&0&\hdots\\

\lambda_1 \bar{\lambda}_2 & \lambda_1 \lambda_2 &0&0 &\bar{\lambda}_1 &  0 & 0&0&0&0&0&0&0& 0&  0& \hdots \\

\lambda_1 \lambda_2 &0 &0& \bar{\lambda}_1 \lambda_2 &0 &\lambda_1  \bar{\lambda}_2& \bar{\lambda}_1  \bar{\lambda}_2   &0& 0& 0& 0& 0& 0& 0&  0& \hdots \\

\lambda_1 \lambda_2 & 0 &0&0 &0&\lambda_1  \bar{\lambda}_2& 0 & \bar{\lambda}_1  \bar{\lambda}_2& \bar{\lambda}_1 \lambda_2    &0& 0&0& 0& 0& 0& \hdots \\

\lambda_1 \bar{\lambda}_2 & \lambda_1 \lambda_2&	0&	0&	0&	0&	0&	0&	\bar{\lambda}_1&	0&	0&	0 &0&  0& 0& \hdots\\

\lambda_1 \lambda_2 &0 &0& \bar{\lambda}_1 \lambda_2 &0&0& 0& 0& 0 &\lambda_1  \bar{\lambda}_2& \bar{\lambda}_1  \bar{\lambda}_2   &0&0&  0& 0& \hdots \\

\lambda_1 \lambda_2 & 0 &0& 0 &0&0 &0&\bar{\lambda}_1  \lambda_2& 0 & \lambda_1  \bar{\lambda}_2&0& \bar{\lambda}_1  \bar{\lambda}_2    &0&  0& 0& \hdots \\

\lambda_1 \lambda_2 & 0 &0& 0 &0&0 &0&0&0&\lambda_1  \bar{\lambda}_2& 0&0 & \bar{\lambda}_1  \bar{\lambda}_2& \bar{\lambda}_1 \lambda_2    & 0& \hdots \\

\lambda_1 \bar{\lambda}_2 & \lambda_1 \lambda_2&	0&	0&	0&	0&	0& 0&	0&	0&	0&	0&	0&	\bar{\lambda}_1&	0&	\hdots \\

\vdots&\vdots&\vdots&\vdots&\vdots&\vdots&\vdots&\vdots&\vdots&\vdots&\vdots&\vdots&\vdots&\vdots&\vdots & \ddots\\ 
\end{bmatrix}
$}
\end{equation}

\end{figure*}

The probability of occurrence of the state $(A,Q,B)$ is denoted by $V_{A,Q,B}$. The steady states vector is $\mathbf{V} = [V_{1,0,0}, V_{1,0,1}, V_{2,0,0}, V_{2,0,1}, V_{2,1,0}, \hdots]$. Hence, by the equations $\mathbf{V}\mathbf{P}=\mathbf{V}$ and $\mathbf{V} \mathbf{1}=1$, we can obtain the steady states. We define $\Delta_{j,k}=\sum_{i=1}^{\infty}V_{i,j,k}$.
The first two equations of $\mathbf{V}\mathbf{P}=\mathbf{V}$ yield
\begin{align}
V_{1,0,0}&=\lambda_1 \lambda_2 \Delta_{0,0} + \lambda_1 \bar{\lambda}_2 \Delta_{0,1} + \lambda_2 \Delta_{1,0}, \ V_{1,0,1}=\lambda_1 \lambda_2 \Delta_{0,1}.  \label{V_101_Delta_01}
\end{align}
Writing the other equations of $\mathbf{V}\mathbf{P}=\mathbf{V}$ can provide us with three recursive equations
\begin{align}
V_{i+1,0,0}&=\bar{\lambda}_1 \bar{\lambda}_2  V_{i,0,0}, \  \forall i \in \{1,2,\hdots\}, \ V_{i+1,0,1}=\bar{\lambda}_1 \lambda_2  V_{i,0,0} + \bar{\lambda}_1 V_{i,0,1}, \  \forall i \in \{1,2,\hdots\},\\
V_{i+1,1,0}&=\lambda_1 \bar{\lambda}_2   V_{i,0,0} + \bar{\lambda}_2 V_{i,1,0}, \  \forall i \in \{1,2,\hdots\}.
\end{align}
Taking $w=\lambda_1 \lambda_2$, $x=\lambda_1 \bar{\lambda}_2$, $y=\bar{\lambda}_1 \lambda_2$, and $z=\bar{\lambda}_1 \bar{\lambda}_2$, by the recursive equations, we can write all of the state probabilities based only on $V_{1,0,0}$ and $V_{1,0,1}$
\begin{align*}
V_{1,0,0}&=V_{1,0,0}, \ V_{1,0,1}=V_{1,0,1}, \ V_{2,0,0}=z V_{1,0,0},\\
V_{2,0,1}&=y V_{1,0,0} + \bar{\lambda}_1 V_{1,0,1}, \ V_{2,1,0}=x V_{1,0,0},\\
V_{3,0,0}&=z^2 V_{1,0,0}, \ V_{3,0,1}=\left(zy+\bar{\lambda}_1 y\right) V_{1,0,0} + \bar{\lambda}_1^2 V_{1,0,1},\\
V_{3,1,0}&=\left(zx+\bar{\lambda}_2 x\right) V_{1,0,0}, \ V_{4,0,0}=z^3 V_{1,0,0},\\
V_{4,0,1}&=\left(yz^2+\bar{\lambda}_1  yz + \bar{\lambda}_1^2 y\right) V_{1,0,0} + \bar{\lambda}_1^3 V_{1,0,1},\\
V_{4,1,0}&=\left(xz^2+\bar{\lambda}_2 xz+\bar{\lambda}_2 x\right) V_{1,0,0}, \cdots \numberthis \label{StateEquations_AoA_Geo/Geo/1/1_Battery}
\end{align*}
If we sum the corresponding terms, we obtain
\begin{equation} \label{Delta_01} 
\Delta_{0,0}=\frac{1}{1-z} V_{1,0,0}, \ \Delta_{0,1}=\frac{y}{\lambda_1} \frac{1}{1-z} V_{1,0,0} + \frac{1}{\lambda_1} V_{1,0,1}, \ \Delta_{1,0}=\frac{x}{\lambda_2} \frac{1}{1-z} V_{1,0,0}.
\end{equation}
By summing and using that the sum of probabilities equals to $1$ ($\mathbf{V} \mathbf{1}=1$), we can obtain one relation and by replacing (\ref{Delta_01}) into (\ref{V_101_Delta_01}) we get another relation between $V_{1,0,1}$ and $V_{1,0,0}$. Solving the two relations yield
\begin{equation} \label{v_100}
V_{1,0,0}=\frac{  \lambda_1 (1-\bar{\lambda}_1 \bar{\lambda}_2) \bar{\lambda}_2 \lambda_2  }{  \bar{\lambda}_1 \lambda_2^3 +\lambda_1 \lambda_2 \bar{\lambda}_2 + \bar{\lambda}_1 \bar{\lambda}_2 \lambda_2^2 + \lambda_1^2 \bar{\lambda}_2^2 },
\end{equation}
and
\begin{equation} \label{v_101}
V_{1,0,1}=\frac{  \lambda_1 \bar{\lambda}_1  \lambda_2^3  }{  \bar{\lambda}_1 \lambda_2^3 +\lambda_1 \lambda_2 \bar{\lambda}_2 + \bar{\lambda}_1 \bar{\lambda}_2 \lambda_2^2 + \lambda_1^2 \bar{\lambda}_2^2 }.
\end{equation}

If we weight the terms in (\ref{StateEquations_AoA_Geo/Geo/1/1_Battery}) corresponding to their AoA, meaning that we multiply terms $V_{A,Q,B}$ by $A$, sum all of the weighted terms, and replace the values of $V_{1,0,0}$ and $V_{1,0,1}$ from (\ref{v_100}) and (\ref{v_101}), and $w$, $x$, $y$, and $z$ we can obtain the expected value of the AoA, i.e., the average AoA as in (\ref{Average_AoA_Battery}).
\hfill \qed

\section{Proof of Theorem \ref{theo:AoAI}} \label{App:AoAI}

\subsection{Proof of Theorem~\ref{theo:AoAI}\ref{theo:The Average AoAI Infinite} and Theorem~\ref{theo:Closed_form}} \label{Proof of The Average AoAI Infinite}
We consider the multi-dimensional Markov chain $(AI,C(AI))$ where $C(.)$ is defined in Definition \ref{def:Queue_state}. Note that here we denote it by $C(AI)$, since if $AI$ is the AoAI, the age of the oldest packet in the queue is at most $AI-1$. Then the states of $(AI,C(AI))$ are
$
\{(1),(2,0),(2,1),(3,0,0),(3,0,1),(3,1,0),(3,1,1),\cdots\}. 
$
\noindent The transition probability matrix is presented in (\ref{TPM_Inf_AoAI}).

\begin{remark}
Following Definition \ref{def:Queue_state}, we indicate the probability of each \textit{queue state} by $\Gamma_{\gamma_{h}\gamma_{h-1}\cdots\gamma_{1}}$, regardless of the AoAI. Each $\Gamma_{\gamma_{h}\gamma_{h-1}\cdots\gamma_{1}}$ captures all the states of the $(AI,C(AI))$. Thus, it is the summation of the probabilities of all the infinite number of states that it includes. For example, $\Gamma_{10}=V_{310}+V_{4010}+V_{50010}+\cdots$, and $\Gamma_{101}=V_{4101}+V_{40101}+V_{400101}+\cdots$. We also define $\Gamma_{0}=V_{1}+V_{20}+V_{300}+\cdots$.

Thus for a unique $\Gamma_{\gamma_{h}\gamma_{h-1}\cdots\gamma_{1}}$, $h$ represents the age of the \textit{oldest packet in the queue}, and $l=\sum_{j=1}^{h}\gamma_j$ is the number of $1$s in $(\gamma_{h}\gamma_{h-1}\cdots\gamma_{1})$ which is in fact the \textit{queue length}.
\end{remark}

Writing the Markov equations from (\ref{TPM_Inf_AoAI}) for the original states $(AI,C(AI))$, we have
\begin{align*}
V_{1}&=w\Gamma_{0}, \ V_{20}=\bar{\lambda}_1 V_{1} + y \Gamma_{1}, \ V_{21}=x V_{1} + w \Gamma_{1}, V_{300}=\bar{\lambda}_1 V_{20} + y \Gamma_{10}, \ V_{301}=x V_{20} + w \Gamma_{10},\\
V_{310}&= z V_{21} + y \Gamma_{11}, \ V_{311}= x V_{21} + w \Gamma_{11}, \cdots
\end{align*}
The number of states related to each $AI$ is $2^{AI-1}$. We define $D_{i}$ as the summation of the probabilities of all the states with $i$ packets in the queue. From the definitions, we know $D_{0}=\Gamma_{0}$. If we sum the states contributing to $D_{i}, \ i \in \{0,1,2,\cdots\}$, and recursively use them for next ones, we ultimately get $D_{i+1}=\left(\frac{x}{y}\right)D_{i}$. We know $\sum_{i=0}^{\infty}D_{i}=D_{0}\frac{1}{1-\frac{x}{y}}=1$. Thus, $D_{0}=1-\frac{x}{y}$, and finally $D_{i}=\left(\frac{x}{y}\right)^{i}\left(1-\frac{x}{y}\right) \ i \in \{0,1,2,\cdots\}$, which is a well-known result (See Remark 1). If we define $f_i$ as
\begin{equation}
f_i=xD_i+wD_{i+1},
\end{equation}
then summing the states contributing to each unique $\Gamma_{\gamma_{j}\gamma_{j-1}\cdots\gamma_{1}}$, we get
\begin{equation*}
\Gamma_{1}=\left(\frac{1}{1-w}\right)\left(xD_0+wzD_1+wyD_2\right),
\end{equation*}
\begin{equation*}
\Gamma_{10}=z\Gamma_{1}+yf_1, \ \Gamma_{11}=x\Gamma_{1}+wf_1,
\end{equation*}
\begin{equation*}
\Gamma_{100}=z\Gamma_{10}+y(zf_1+yf_2), \ \Gamma_{101}=x\Gamma_{10}+w(zf_1+yf_2),
\end{equation*}
\begin{equation*}
\Gamma_{110}=z\Gamma_{11}+y(xf_1+wf_2), \ \Gamma_{111}=x\Gamma_{11}+w(xf_1+wf_2),
\end{equation*}
\begin{equation*}
\Gamma_{1000}=z\Gamma_{100}+y(z(zf_1+yf_2)+y(zf_2+yf_3)),
\end{equation*}
\begin{equation*}
\Gamma_{1001}=x\Gamma_{100}+w(z(zf_1+yf_2)+y(zf_2+yf_3)),
\end{equation*}
\begin{equation*}
\Gamma_{1010}=z\Gamma_{101}+y(x(zf_1+yf_2)+w(zf_2+yf_3)),
\end{equation*}
\begin{equation*}
\Gamma_{1011}=x\Gamma_{101}+w(x(zf_1+yf_2)+w(zf_2+yf_3)),
\end{equation*}
\begin{equation*}
\Gamma_{1100}=z\Gamma_{110}+y(z(xf_1+wf_2)+y(xf_2+wf_3)),
\end{equation*}
\begin{equation*}
\Gamma_{1101}=x\Gamma_{110}+w(z(xf_1+wf_2)+y(xf_2+wf_3)),
\end{equation*}
\begin{equation*}
\Gamma_{1110}=z\Gamma_{111}+y(x(xf_1+wf_2)+w(xf_2+wf_3)),
\end{equation*}
\begin{equation*}
\Gamma_{1111}=x\Gamma_{111}+w(x(xf_1+wf_2)+w(xf_2+wf_3)),
\end{equation*}
and so on. If we make multiplications and simplifications, we get that for a given $h$, the probabilities of the ones with the same $l$ are the same:
\begin{equation*}
\Gamma_{10}=z\Gamma_{1}+yf_1, \ \Gamma_{11}=x\Gamma_{1}+wf_1,
\
\Gamma_{100}=z^2 \Gamma_{1}+ 2yzf_1 + y^2f_2, 
\end{equation*}
\begin{equation*}
\Gamma_{101}=\Gamma_{110}=xz \Gamma_{1}+ (wz+xy) f_1+ wyf_2, 
\
\Gamma_{111}=x^2 \Gamma_{1}+ 2wxf_1 + w^2f_2,
\end{equation*}
\begin{equation*}
\Gamma_{1000}=z^3 \Gamma_{1} + 3yz^2f_1 + 3y^2zf_2 + y^3f_3,
\end{equation*}
\begin{align*}
\Gamma_{1001}=\Gamma_{1010}=\Gamma_{1100}= xz^2\Gamma_{1}+(wz^2+2xyz)f_1 + (2wyz+xy^2)f_2 + wy^2f_3,
\end{align*}
\begin{align*}
\Gamma_{1011}=\Gamma_{1101}=\Gamma_{1110}=x^2z\Gamma_{1} + (x^2y+2wxz)f_1 + (w^2z+2wxy)f_2 + w^2yf_3,
\end{align*}
\begin{equation*}
\Gamma_{1111}=x^3\Gamma_{1}+ 3wx^2f_1 + 3w^2xf_2 + w^3f_3, \cdots
\end{equation*}
The $\Gamma$s for different $l$s as a function of $h$ are:
\begin{equation*}
\Gamma(h,1)=z^{h-1}\Gamma_1+\sum_{k=1}^{\infty} {h-1 \choose k} y^{k} z^{h-1-k} f_{k},
\end{equation*}
\begin{align*}
\Gamma(h,2)=xz^{h-2}\Gamma_1+\sum_{k=0}^{\infty} {h-2 \choose k} w y^{k} z^{h-2-k} f_{k+1} + \sum_{k=1}^{\infty} {h-2 \choose k} x y^{k} z^{h-2-k} f_{k},
\end{align*}
\begin{align*}
\Gamma(h,3)=x^2 z^{h-3}\Gamma_1 +& \sum_{k=0}^{\infty}{h-3 \choose k}w^{2}y^{k}z^{h-3-k}f_{k+2}+ 2 {h-3 \choose k} wxy^{k}z^{h-3-k} f_{k+1}\\
&+ \sum_{k=1}^{\infty} {h-3 \choose k}x^2 y^{k}z^{h-3-k}f_{k}, \cdots
\end{align*}
Then, the general pattern as the closed-form formula of $\Gamma(h,l)$ is (\ref{Final_Closed_Form}).
\hfill \qed

We define $\Gamma_{h:}$ as the summation of all the $\Gamma_{\Phi_{j}\Phi_{j-1}\cdots\Phi_{1}}$s with the length $h$ of $\Phi_{j}\Phi_{j-1}\cdots\Phi_{1}$. Then
\begin{equation}
\Gamma_{h:}=\bar{\lambda}_2^{h-1}\Gamma_1+\sum_{j=1}^{h-1} {h-1 \choose j} \lambda_2^j \bar{\lambda}_2^{h-j-1} f_j, \ \forall h \in \mathbb{N}
\end{equation}

Then, summing corresponding terms from the states we get
\begin{equation}
AI_1=w\Gamma_0
\end{equation}
\begin{equation}
AI_{i}=yV_{i-1} + \bar{\lambda}_2 AI_{i-1} + \lambda_2 \Gamma_{i-1:}, \ i=2,3,\cdots
\end{equation}

Then, substituting from those, we get
\[
\begin{aligned}
AI_1 &= V_1, \ AI_2 = V_1 \left(y + \bar{\lambda}_2\right)  + \lambda_2 \Gamma_1, AI_3 = V_1 \left(y\left(\bar{\lambda}_1 + \bar{\lambda}_2 \right) + \bar{\lambda}_2^2\right) + \Gamma_1 \left(y^2 + 2\lambda_2 \bar{\lambda}_2\right) + \lambda_2^2 f_1, \\
AI_4 &= V_1 \left(y\left(\bar{\lambda}_1^2  + \bar{\lambda}_1 \bar{\lambda}_2  + \bar{\lambda}_2^2 \right) + \bar{\lambda}_2^3\right) + \Gamma_1 \left(y^2\left(\bar{\lambda}_1   + \bar{\lambda}_2 +z\right) + 3 \lambda_2 \bar{\lambda}_2^2\right) + f_1 \left(y^3 + 3 \lambda_2^2 \bar{\lambda}_2\right) + f_2 \lambda_2^3,\\
AI_5 &= V_1 \left(y\left(\bar{\lambda}_1^3  + \bar{\lambda}_1^2 \bar{\lambda}_2  + \bar{\lambda}_1\bar{\lambda}_2^2 +\bar{\lambda}_2^3 \right) + \bar{\lambda}_2^4\right) + \Gamma_1 \left(y^2\left(\bar{\lambda}_1^2 + \bar{\lambda}_1 \bar{\lambda}_2 + \bar{\lambda}_2^2 +  z \left(\bar{\lambda}_1+\bar{\lambda}_2\right) + z^2 \right) + 4 \lambda_2 \bar{\lambda}_2^3\right) +\\
&f_1 \left(y^3 \left(\bar{\lambda}_1 + \bar{\lambda}_2 +2z\right) + 6 \lambda_2^2 \bar{\lambda}_2^2\right) + f_2\left( y^4 + 4\lambda_2^3\bar{\lambda}_2\right) + f_3 \lambda_2^4, \cdots
\end{aligned}%
\]

Weighting the terms and summing them and substituting the $V_1$, $\Gamma_1$, and $f_i$s, and also $w$, $x$, $y$, and $z$, we get 
\begin{equation} \label{AoAI_average_InfCache_Appendix}
\overline{AI}=\frac{\lambda_1^2 (\lambda_2-1) (\lambda_1 - \lambda_2) + \lambda_2^3( \lambda_1-1)}{\lambda_1 (\lambda_1 - \lambda_2) \lambda_2^2}
\end{equation}
\hfill \qed

\subsection{Proof of Theorem \ref{theo:AoAI}\ref{theo:The Average AoAI cache}} \label{Proof of The Average AoAI cache}
Consider the two-dimensional Markov chain $(AI,I)$, in which $AI \in \{1,2,3,\hdots\}$ is the AoAI and $I \in \{1,2,3,\hdots \}$ is the AoI, at the end of each time slot. Thus, the states are $\{1,1\}$, $\{2,1\}$, $\{2,2\}$, $\{3,1\}$, $\{3,2\}$, $\{3,3\}$, $\cdots$ . Obviously, we cannot have states such that $AI<I$. Then, the transition probability matrix for this Markov chain is (\ref{TPM_1_AoAI}).

The first equation is
\begin{equation} \label{V_11}
V_{1,1}=\lambda_1 \lambda_2.
\end{equation} 
Other equations yield patterns as
\begin{equation} \label{V_AI1}
V_{AI+1,1}= \lambda_1  \bar{\lambda}_2 \sum_{j=1}^{AI} V_{AI,j}  \ , AI \geq 1,
\end{equation}
\begin{equation} \label{V_AI1equalI1}
V_{AI+1,I+1}=\bar{\lambda}_1   V_{AI,I} + \bar{\lambda}_1  \lambda_2 \sum_{i=AI+1}^{\infty} V_{i,I} \ , AI=I \geq 1,
\end{equation}
\begin{equation} \label{V_AI1notequalI1}
V_{AI+1,I+1}=\bar{\lambda}_1  \bar{\lambda}_2 V_{AI,I}  \ , AI \neq I \geq 1 .
\end{equation}
Summing the equations (\ref{V_11}) and  (\ref{V_AI1}) for all $AI$s yield
\begin{equation} \label{I_1}
I_1 = \lambda_1 .
\end{equation}
By using (\ref{I_1}) (in structures of (\ref{V_AI1equalI1})) and all the equations and structures (\ref{V_11}), (\ref{V_AI1}), (\ref{V_AI1equalI1}), and (\ref{V_AI1notequalI1}) we can recursively write all the states.

Then, noting $\lambda_1-y=z$, the steady states can be written and simplified against $\lambda_1$ and $\lambda_2$.
\begin{align*}
V_{1,1}=&w, \ V_{2,1}=w(x), \ V_{2,2}=w(z)+\lambda_1 y, \ V_{3,1}=w(x^2 + xz ) + \lambda_1 y (x), \ V_{3,2}=w(xz) \\
V_{3,3}=&w(z^2)+\lambda_1 y( \bar{\lambda}_1 + z) \ , V_{4,1}=w(x^3 +  2x^2z +  x z^2 )+ \lambda_1 y ( x^2 + \bar{\lambda}_1 x +  x z)  \\
V_{4,2}=&w(x^2z+ xz^2) + \lambda_1 y (xz), \ V_{4,3}=w(xz^2),  \ V_{4,4}=w( z^3)+\lambda_1 y ( \bar{\lambda}_1^2  +  \bar{\lambda}_1 z + z^2 ), \cdots
\numberthis \label{StateEquations_AoAI_Geo/Geo/1/1_Controller}
\end{align*}
Since $V_{1,1}$ is explicitly obtained, we straightly proceed to obtain the average AoAI. To do so, we weight the states equations (\ref{StateEquations_AoAI_Geo/Geo/1/1_Controller}) by their AoAI, i.e. multiply $V_{AI,I}$ by $AI$, and sum them together. One can see that the summation can be stated as
\begin{equation}
\overline{AI}=\sum_{i=1}^{\infty} i V_{i,i} + \sum_{i=2}^{\infty} \left( \frac{i-1}{1-z} +\frac{1}{(1-z)^2} \right) V_{i,1}.
\end{equation}

Summing all the terms together, we obtain the average AoAI for the system model in this work as
\begin{equation}
\overline{AI}=\frac{1}{\lambda_1}+\frac{1}{\lambda_2}-1  \hfill \qed
\end{equation}

\subsection{Proof of Theorem~\ref{theo:AoAI}\ref{theo:The Average AoAI battery}} \label{Proof of The Average AoAI battery}
Let $(AI,I,B)$ be a three-dimensional Markov chain, in which $AI \in \{1,2,3,\hdots\}$ is the AoAI at the end of time slot $n$, $B \in \{0,1\}$ is the battery being empty or not, with $0$ representing an empty battery and $1$ a full battery at the end of time slot $n$, and $I \in \{1,2,3,\hdots \}$ is the AoI at the end of time slot $n$. Thus, the states are $\{1,1,0\}$, $\{1,1,1\}$, $\{2,1,0\}$, $\{2,2,0\}$, $\{2,2,1\}$, $\{3,1,0\}$, $\{3,2,0\}$, $\{3,3,0\}$, $\{3,3,1\}$, $\{4,1,0\}$, $\cdots$. Obviously, we cannot have states such that $AI<I$. Also, we cannot have states such that $B=1, I<AI$, since if there was enough energy, then the packet in the queue should have been utilized and the AoAI and AoI should have been the same, i.e., $B=1$ only if $AI=I$. The transition probability matrix is given by (\ref{TPM_Battery_AoAI}).

The first two Markov equations are
\begin{equation} \label{v_110andv_111}
V_{1,1,0}=\lambda_1 \lambda_2 B_0 + \lambda_1 \bar{\lambda}_2 B_1 , \ V_{1,1,1}=\lambda_1 \lambda_2 B_1,  
\end{equation}
in which $B_1=\sum_{i=1}^{\infty}V_{i,i,1}$ and $B_0=1-B_1$. Other equations yield patterns as
\begin{equation} \label{V_AI10}
V_{AI,1,0}=\lambda_1 \bar{\lambda}_2 \sum_{I=1}^{AI-1} V_{AI-1,I,0},
\end{equation}
\begin{equation} \label{V_AII0}
V_{AI,I,0}= \bar{\lambda}_1  \bar{\lambda}_2 V_{AI-1,I-1,0}  \ , I \neq 1, I \neq AI,
\end{equation}
\begin{equation} \label{V_AII0_2}
V_{AI,I,0}=\bar{\lambda}_1  \bar{\lambda}_2 V_{AI-1,I-1,0} + \bar{\lambda}_1  \lambda_2 \sum_{i=AI}^{\infty} V_{i,I-1,0} \ , AI=I \geq 2,
\end{equation}
\begin{equation} \label{V_AII1}
V_{AI,I,1}=\bar{\lambda}_1  \lambda_2 V_{AI-1,I-1,0} + \bar{\lambda}_1  V_{AI-1,I-1,1} \ , AI=I \geq 2.
\end{equation}
Then $V_{1,1,0}$ and $V_{1,1,1}$ from (\ref{v_110andv_111}) yield
\begin{equation} \label{AI_1}
AI_1=\lambda_1 \lambda_2 + \lambda_1 \bar{\lambda}_2 B_1.
\end{equation}
Summing over different values of $AI$, for (\ref{V_AI10}) we obtain
\begin{equation} \label{I_1_A_1}
I_1 -AI_1 = \lambda_1 \bar{\lambda}_2 B_0.
\end{equation}
From (\ref{AI_1}) and (\ref{I_1_A_1}) we get $I_1=\lambda_1$, and by employing it in structures of (\ref{V_AII0_2}) and using all the structures (\ref{V_AI10}), (\ref{V_AII0}), (\ref{V_AII0_2}), and (\ref{V_AII1}) we can write
\begin{align*}
V_{1,1,0}=&V_{1,1,0}, \ V_{1,1,1}=V_{1,1,1}, \ V_{2,1,0}=x V_{1,1,0}, \ V_{2,2,0}=\left(z-y\right) V_{1,1,0} + \left(-y\right) V_{1,1,1}+y \lambda_1,\\
V_{2,2,1}=&y V_{1,1,0} + \bar{\lambda}_1 V_{1,1,1}, \ V_{3,1,0}=\left(x^2 + xz-xy\right) V_{1,1,0}+\left(-xy\right)V_{1,1,1}+xy\lambda_1, \\
V_{3,2,0}=&xz V_{1,1,0}, \ V_{3,3,0}=\left(z^2 -2yz\right) V_{1,1,0} + \left(-2yz\right)V_{1,1,1} + 2yz\lambda_1,\\
V_{3,3,1}=&\left(yz -y^2 +\bar{\lambda}_1 y \right) V_{1,1,0} + \left(-y^2 + \bar{\lambda}_1^2\right)V_{1,1,1} + y^2 \lambda_1, \cdots  \numberthis \label{StateEquations_AoAI_Geo/Geo/1/1_Battery}
\end{align*}
Summing all the states and equaling them to $1$ provides one relation between $V_{1,1,0}$ and $V_{1,1,1}$ and the other relation can be obtained from $V_{1,1,1}$ in (\ref{v_110andv_111}) by summing only the terms of which $B_1$ consists.
Having these two relations, we obtain 
\begin{equation} \label{v_110_}
V_{1,1,0}=\frac{\lambda_1 (\lambda_1^2  \bar{\lambda}_2 +\lambda_2) \bar{\lambda}_2 \lambda_2}{
\lambda_1^2 \bar{\lambda}_2^2 + \lambda_2^2 + \lambda_1 \lambda_2(1 - 2 \lambda_2)},
\end{equation}
and
\begin{equation} \label{v_111_}
V_{1,1,1}= \frac{\bar{\lambda}_1 \lambda_1 \lambda_2^3}{
\lambda_1^2 \bar{\lambda}_2^2 + \lambda_2^2 + \lambda_1 \lambda_2(1 - 2 \lambda_2)}.
\end{equation}
If we multiply the states equations (\ref{StateEquations_AoAI_Geo/Geo/1/1_Battery}) by their AoAI, i.e. multiply $V_{AI,I,B}$ by $AI$, and sum them together, and replace (\ref{v_110_}) and (\ref{v_111_}) and also $w$, $x$, $y$, and $z$, and simplify, we obtain the average AoAI (\ref{Average_AoAI_Battery}).
\hfill \qed
\end{appendices}

\bibliographystyle{ieeetr}
\bibliography{bibliography.bib}

@ARTICLE{kountouris2021semantics,
  author={Kountouris, Marios and Pappas, Nikolaos},
  journal={IEEE Comm. Mag.}, 
  title={Semantics-Empowered Communication for Networked Intelligent Systems}, 
  year={2021}}

@ARTICLE{gunduz2022beyond,
  author={Gündüz, Deniz and Qin, Zhijin and Aguerri, Inaki Estella and Dhillon, Harpreet S. and Yang, Zhaohui and Yener, Aylin and Wong, Kai Kit and Chae, Chan-Byoung},
  journal={IEEE Journal on Selected Areas in Comm.}, 
  title={Beyond Transmitting Bits: Context, Semantics, and Task-Oriented Communications}, 
  year={2023}}

@ARTICLE{chen2021optimization,
  author={Chen, Zheng and Pappas, Nikolaos and Björnson, Emil and Larsson, Erik G.},
  journal={IEEE Open J. of the Comm. Society}, 
  title={Optimizing Information Freshness in a Multiple Access Channel With Heterogeneous Devices}, 
  year={2021}}

@ARTICLE{yates2021age,
  author={Yates, Roy D. and Sun, Yin and Brown, D. Richard and Kaul, Sanjit K. and Modiano, Eytan and Ulukus, Sennur},
  journal={IEEE J. on Selected Areas in Comm.}, 
  title={Age of Information: An Introduction and Survey}, 
  year={2021}}

@ARTICLE{hatami2021aoi,
  author={Hatami, Mohammad and Leinonen, Markus and Codreanu, Marian},
  journal={IEEE T. on Comm.}, 
  title={AoI Minimization in Status Update Control With Energy Harvesting Sensors}, 
  year={2021}}

@ARTICLE{hatami2022on,
  author={Hatami, Mohammad and Leinonen, Markus and Chen, Zheng and Pappas, Nikolaos and Codreanu, Marian},
  journal={IEEE T. on Comm.}, 
  title={On-Demand AoI Minimization in Resource-Constrained Cache-Enabled IoT Networks With Energy Harvesting Sensors}, 
  year={2022}}

@ARTICLE{zheng2019closed,
  author={Zheng, Xi and Zhou, Sheng and Jiang, Zhiyuan and Niu, Zhisheng},
  journal={IEEE T. on Wireless Comm.}, 
  title={Closed-Form Analysis of Non-Linear Age of Information in Status Updates With an Energy Harvesting Transmitter}, 
  year={2019}}

@ARTICLE{arafa2019age,
  author={Arafa, Ahmed and Yang, Jing and Ulukus, Sennur and Poor, H. Vincent},
  journal={IEEE T. on Information Theory}, 
  title={Age-Minimal Transmission for Energy Harvesting Sensors With Finite Batteries: Online Policies}, 
  year={2020}}

@ARTICLE{elmagid2022age,
  author={Abd-Elmagid, Mohamed A. and Dhillon, Harpreet S.},
  journal={IEEE J. on Selected Areas in Information Theory}, 
  title={Age of Information in Multi-source Updating Systems Powered by Energy Harvesting}, 
  year={2022}}

@ARTICLE{bacingolu2019optimal,
  author={Bacinoglu, Baran Tan and Sun, Yin and Uysal, Elif and Mutlu, Volkan},
  journal={J. of Comm. and Networks}, 
  title={Optimal status updating with a finite-battery energy harvesting source}, 
  year={2019}}

@ARTICLE{krikidis2019average,
  author={Krikidis, Ioannis},
  journal={IEEE Wireless Comm. Letters}, 
  title={Average Age of Information in Wireless Powered Sensor Networks}, 
  year={2019}}

@ARTICLE{ibrahim2016stability,
  author={Ibrahim, Abdelrahman M. and Ercetin, Ozgur and ElBatt, Tamer},
  journal={IEEE J. on Selected Areas in Comm.}, 
  title={Stability Analysis of Slotted Aloha With Opportunistic RF Energy Harvesting}, 
  year={2016}}

@book{srikant2013communication,
  title={Communication networks: an optimization, control, and stochastic networks perspective},
  author={Srikant, Rayadurgam and Ying, Lei},
  year={2013},
  publisher={Cambridge University Press}
}

@ARTICLE{abdelmagid2020aoi,
  author={Abd-Elmagid, Mohamed A. and Dhillon, Harpreet S. and Pappas, Nikolaos},
  journal={IEEE T. on Vehicular Tech.}, 
  title={AoI-Optimal Joint Sampling and Updating for Wireless Powered Communication Systems}, 
  year={2020}}

@ARTICLE{delfani2025semantics,
  author={Delfani, Erfan and Pappas, Nikolaos},
  journal={IEEE Comm. Letters}, 
  title={Semantics-Aware Updates From Remote Energy Harvesting Devices to Interconnected LEO Satellites}, 
  year={2025}}

@book{pappas2023age,
  title={Age of Information: Foundations and Applications},
  author={Pappas, Nikolaos and Abd-Elmagid, Mohamed A and Zhou, Bo and Saad, Walid and Dhillon, Harpreet S},
  year={2023},
  publisher={Cambridge University Press}
}

@ARTICLE{wu2017optimal,
  author={Wu, Xianwen and Yang, Jing and Wu, Jingxian},
  journal={IEEE T. on Green Comm. and Networking},
  title={Optimal Status Update for Age of Information Minimization With an Energy Harvesting Source}, 
  year={2018}}

@ARTICLE{gindullina2021age,
  author={Gindullina, Elvina and Badia, Leonardo and Gündüz, Deniz},
  journal={IEEE T. on Green Comm. and Networking},
  title={Age-of-Information With Information Source Diversity in an Energy Harvesting System}, 
  year={2021}}

@INPROCEEDINGS{nikkhah2023age,
	author={Nikkhah, Ali and Ephremides, Anthony and Pappas, Nikolaos},
	booktitle={IEEE INFOCOM WKSHPS}, 
	title={Age of Actuation in a Wireless Power Transfer System}, 
	year={2023}}

@article{nikkhah2023ageo,
  title={Age of Actuation and Timeliness: Semantics in a Wireless Power Transfer System},
  author={Nikkhah, Ali and Ephremides, Anthony and Pappas, Nikolaos},
  journal={arXiv preprint arXiv:2312.13919},
  year={2023}
}

@ARTICLE{xu2023optimal,
  author={Xu, Chao and Zhang, Xinyan and Yang, Howard H. and Wang, Xijun and Pappas, Nikolaos and Niyato, Dusit and Quek, Tony Q. S.},
  journal={IEEE T. on Mobile Computing}, 
  title={Optimal Status Updates for Minimizing Age of Correlated Information in IoT Networks With Energy Harvesting Sensors}, 
  year={2023}}

@ARTICLE{sudevalayam2011energy,
  author={Sudevalayam, Sujesha and Kulkarni, Purushottam},
  journal={IEEE Comm. Surveys and Tutorials}, 
  title={Energy Harvesting Sensor Nodes: Survey and Implications}, 
  year={2011}}

@ARTICLE{lu2023semantics,
  author={Lu, Zhilin and Li, Rongpeng and Lu, Kun and Chen, Xianfu and Hossain, Ekram and Zhao, Zhifeng and Zhang, Honggang},
  journal={IEEE Comm. Surveys \& Tutorials}, 
  title={Semantics-Empowered Communications: A Tutorial-cum-Survey}, 
  year={2023}}

@ARTICLE{feng2021age,
  author={Feng, Songtao and Yang, Jing},
  journal={IEEE T. on Comm.}, 
  title={Age of Information Minimization for an Energy Harvesting Source With Updating Erasures: Without and With Feedback}, 
  year={2021}}

@ARTICLE{zhao2025age,
  author={Zhao, Fangming and Pappas, Nikolaos and Zhang, Meng and Yang, Howard H.},
  journal={IEEE J. on Selected Areas in Comm.}, 
  title={Age of Information in Random Access Networks With Energy Harvesting}, 
  year={2025}}

@ARTICLE{jia2021age,
  author={Jia, Xiangdong and Cao, Shengnan and Xie, Mangang},
  journal={IEEE Wireless Comm. Letters}, 
  title={Age of Information of Dual-Sensor Information Update System With HARQ Chase Combining and Energy Harvesting Diversity}, 
  year={2021}}

@INPROCEEDINGS{chang2020age,
  author={Chang, Bo and Li, Liying and Zhao, Guodong and Meng, Zhen and Imran, Muhammad Ali and Chen, Zhi},
  booktitle={IEEE INFOCOM WKSHPS}, 
  title={Age of Information for Actuation Update in Real-Time Wireless Control Systems}, 
  year={2020}}

@ARTICLE{zhao2019toward,
  author={Zhao, Guodong and Imran, Muhammad Ali and Pang, Zhibo and Chen, Zhi and Li, Liying},
  journal={IEEE Comm. Mag.}, 
  title={Toward Real-Time Control in Future Wireless Networks: Communication-Control Co-Design}, 
  year={2019}}

@Article{kyung2024priority,
AUTHOR = {Kyung, Yeunwoong and Sung, Jihoon and Ko, Haneul and Song, Taewon and Kim, Youngjun},
TITLE = {Priority-Aware Actuation Update Scheme in Heterogeneous Industrial Networks},
JOURNAL = {Sensors},
VOLUME = {24},
YEAR = {2024},
NUMBER = {2},
ARTICLE-NUMBER = {357}
}

@article{chang2021effective,
  author    = {Bo Chang and Burak Kizilkaya and Liying Li and Guodong Zhao and Zhi Chen and Muhammad Ali Imran},
  title     = {Effective age of information in real-time wireless feedback control systems},
  journal   = {Science China Information Sciences},
  year      = {2021}}

@INPROCEEDINGS{champatiperformance2019,
  author={Champati, Jaya Prakash and Mamduhi, Mohammad H. and Johansson, Karl H. and Gross, James},
  booktitle={IEEE INFOCOM WKSHPS}, 
  title={Performance Characterization Using AoI in a Single-loop Networked Control System}, 
  year={2019}}

@misc{luo2025informationfreshnesssemanticsinformation,
      title={{From Information Freshness to Semantics of Information and Goal-oriented Communications}}, 
      author={Jiping Luo and Erfan Delfani and Mehrdad Salimnejad and Nikolaos Pappas},
      year={arXiv: 2512.12758, 2025},
}

@INPROCEEDINGS{nikkhah2024aoai,
  author={Nikkhah, Ali and Ephremides, Anthony and Pappas, Nikolaos},
  booktitle={IEEE GLOBECOM}, 
  title={Age of Actuated Information and Age of Actuation in a Data-Caching Energy Harvesting Actuator}, 
  year={2024}}

@ARTICLE{li2026aoi,
  author={Li, Shuang and Hu, Huimin and Yang, Hong-Chuan and Xiong, Ke and Fan, Pingyi and Ben Letaief, Khaled},
  journal={IEEE T. on Cognitive Comm. and Networking}, 
  title={AoI Minimization for WP-IoT With PDQN-Based Hybrid Offline/Online Learning: A Joint Scheduling and Transmission Design Approach}, 
  year={2026}}

@INPROCEEDINGS{abdelmagid2022distribution,
  author={Abd-Elmagid, Mohamed A. and Dhillon, Harpreet S.},
  booktitle={IEEE INFOCOM WKSHPS}, 
  title={Distribution of AoI in EH-powered Multi-source Systems under Non-preemptive and Preemptive Policies}, 
  year={2022}}

@INPROCEEDINGS{xiao2024infromation,
  author={Xiao, Shuyu and Sun, Xinghua and Zhan, Wen and Wang, Xijun},
  booktitle={IEEE ITW}, 
  title={Information Freshness in Random Access Networks with Energy Harvesting}, 
  year={2024}}

@ARTICLE{ngo2025timely,
  author={Ngo, Khac-Hoang and Durisi, Giuseppe and Munari, Andrea and Lázaro, Francisco and Amat, Alexandre Graell i},
  journal={IEEE T. on Comm.}, 
  title={Timely Status Updates in Slotted ALOHA Networks With Energy Harvesting}, 
  year={2025}}

@INPROCEEDINGS{jaiswal2023age,
  author={Jaiswal, Akanksha and Chattopadhyay, Arpan},
  booktitle={WiOpt}, 
  title={Age-of-Information Minimization for Energy Harvesting Sensor in Non-Stationary Environment}, 
  year={2023}}

@INPROCEEDINGS{zhang2025aoi,
  author={Zhang, Xingyuan and Jia, Xiangdong and Bao, Hongli and Wu, Jingjing},
  booktitle={IEEE VTC}, 
  title={AoI-Optimized Scheduling for Wireless-Powered Cognitive Radio Networks with Short-Packet Communication}, 
  year={2025}}

@INPROCEEDINGS{banerjee2024minimizing,
  author={Banerjee, Subhankar and Ulukus, Sennur},
  booktitle={WiOpt}, 
  title={Minimizing Age of Information in an Energy-Harvesting Scheduler with Rateless Codes}, 
  year={2024}}

\end{document}